%% file: main.tex
\documentclass[11pt]{article}
\usepackage[margin=1in]{geometry}
\usepackage{lmodern}
\usepackage{authblk}
\usepackage{amsthm}
\usepackage[utf8]{inputenc}
\usepackage[T1]{fontenc}
\usepackage[english]{babel}
\usepackage{graphicx}
\usepackage{xcolor}
\usepackage{float}
\usepackage{flafter}
\usepackage{hyperref}

\usepackage{subfiles}
\usepackage{subcaption}
\usepackage{wrapfig}
\usepackage{comment}
\usepackage{amsmath}
\usepackage{amssymb}
\usepackage[authoryear]{natbib}
\usepackage{changepage}
\usepackage{array}
\usepackage{booktabs}
\usepackage{tikz}
\usetikzlibrary{shapes.geometric, arrows.meta, positioning, fit, backgrounds}
\theoremstyle{plain}
\newtheorem{theorem}{Theorem}
\newtheorem{proposition}{Proposition}
\theoremstyle{definition}

\hypersetup{colorlinks=true, linkcolor=blue!55!black, citecolor=blue!55!black, urlcolor=blue!55!black}

\setcitestyle{round,semicolon}

\title{Concave Processes for Multi-Task Career Trajectories: Evaluating Aging Across Multiple Measures of NBA Performance}
\author[1]{Abhijit Brahme\thanks{\texttt{abhijitbrahme@ucsb.edu}}}
\author[1]{Alexander Franks\thanks{\texttt{amfranks@ucsb.edu}}}
\affil[1]{Department of Statistics and Applied Probability, University of California, Santa Barbara, Santa Barbara, CA 93106}
\date{}

\begin{document}
\maketitle

\begin{abstract}
NBA athlete performance tends to increase through early career as athletes develop and acclimate to the league, followed by decline due to age-related deterioration in athleticism. While this general pattern persists, the precise shape of this trajectory varies by athlete and across different measures of performance. To model performance increase and decline, we introduce the concave process prior,
a novel nonparametric prior over concave functions. We then use a latent variable model to characterize dependence in aging profiles across player-metrics, embedding each player in a
shared low-dimensional latent space so that players with similar profiles learn
similar trajectory shapes, peak ages, and peak values. Posterior analysis of the learned embedding supports latent-space nearest-neighbor
retrieval of career-comparable players and informed projections of young players. We apply our model to data across over a dozen performance
metrics for over two thousand players in seasons ranging from 1997 to 2026. Our results show that jointly modeling all metrics improves held-out predictive performance over single-metric alternatives, and that the concavity constraint itself improves prediction.  We find that athleticism-driven metrics such as blocks and offensive rebounds peak in a player's early twenties, while skill-based shooting metrics peak in the mid-twenties or later.
\end{abstract}

\section{Introduction}

Performance across all demanding human endeavors follow the same performance pattern.  Early on, performance improves as skills are acquired: extended, structured practice produces large and reliably measured gains in domains ranging from music to chess to sport \citep{ericsson1993deliberate, macnamara2014deliberate}.  Later, the physiological and cognitive capacities that support execution begin to erode.  Physically, muscle mass, strength, speed, and aerobic power decline after forty years of age onward and accelerates with age \citep{tanaka2008endurance}.  On the cognitive side, measurable decline in processing speed, reasoning, and memory begins well before mid-life \citep{salthouse2009when}.  Together, the two processes yield a single-peaked aging profile documented for creative and professional achievement \citep{simonton1988age}.  

In athletic competition, such aging patterns are often particularly important since performance is measured over the course of athletes complete careers. Production curves have been estimated across multiple sports ranging from baseball to hockey to golf \citep{berry1999bridging, schulz1994relationship, bradbury2009peak}.
In these studies, production is approximately unimodal in age.  Within a sport, the age of peak is fairly stable, although the timing of this peak, and the rate of performance decline, vary by sports and are influenced by the  physical and cognitive demands unique to the task \citep{schulz1988peak, allen2015age, berry1999bridging, fair2007estimated}.   This evidence sorts abilities into two broad families that age in opposite ways: skill-based abilities, accumulated through sustained, structured practice, that remain comparatively robust once acquired, and athleticism-based abilities that atrophy earlier and fall away more steeply.

 Basketball is a sport where many observable skills fall into a blend of both regimes.  A single roster spot demands shooting accuracy, rebounding, playmaking, and defense at once, and these draw on a combination of learned skill and raw athleticism; analyses of professional basketball players find evidence of exactly this kind of trait compensation with age \citep{lailvaux2014trait}.  General managers must project whether a player’s production will hold over the life of a proposed contract.  For example, the Brooklyn Nets’ 2013 trade for aging stars Paul Pierce and Kevin Garnett (35 and 37 years old respectively) in exchange for future draft picks and younger prospects illustrates the franchise-altering cost of misjudging career trajectories.  Paul Pierce retired shortly after the end of that season, Kevin Garnett was traded in middle of the subsequent season and without draft capital or young talent the Nets quickly became one of the worst teams in the league. From the athlete’s perspective, knowing which skills decline earliest can guide training adjustments: players may shift toward shooting and decision-making, relying on developing skills rather than athleticism to impact the game.  \textit{Production curve analysis} formalizes this challenge as inferring smooth trajectories $f_p(t)$ from sparse longitudinal observations across players \citep{wakim2014functional}.

The bulk of the production curve literature models focus on modeling the evolution of a scalar summary of NBA performance \citep{page2007using, page2013effect, vaci2019large, wakim2014functional}, and ignore more nuanced multivariate nature of player performance across multiple aspects of play like shooting, rebounding, and defensive ability.  Rating systems such as RAPTOR \citep[Robust Algorithm using Player Tracking and On/Off Ratings;][]{natesilver538_2015, natesilver538_2019} estimate a single summary score via nearest-neighbor averaging of similar historical players. They impose similarity through fixed feature comparisons and are susceptible to selection bias at career extremes.  \citet{vaci2019large} improve upon naive averages by fitting piecewise parametric curves with the peak location specified a priori, but the rigid functional form limits adaptability to diverse career shapes.  Methods that do incorporate player pooling either use categorical positional types \citep{page2007using} or pre-specified continuous similarity \citep{wakim2014functional, vinue2019forecasting}; in particular, \citet{vinue2019forecasting} represent each player as a convex combination of a fixed set of archetypoids \citep{vinue2015archetypoids}, constraining the representation to the convex hull of an a-priori determined archetype set.

\subsection{Contributions}
\label{sec:contributions}

Our analysis is organized around three questions about how professional athletes age.

\begin{enumerate}
\item \textbf{How does production evolve over a career, and how does that evolution vary across athletes?}  A career is characterized not only by average performance but also by \emph{when} players peak and how steeply they decline away afterwards.  We therefore want the age of peak performance and the value attained at that peak to be explicit, estimable quantities carrying their own  uncertainty, so that athletes may be compared on the shape of their careers rather than on a single career summary.

\item \textbf{Do different aspects of performance follow different aging trajectories?}  Shooting accuracy, rebounding, playmaking, and defense place different demands on learned skill and on physical capacity. These qualities do not necessarily peak together within a player or across the league.  Answering this requires estimating a distinct trajectory for each metric while still pooling strength across metrics and players.

\item \textbf{Does multi-task inference improve forecasts of player skill?}  Decisions about contracts, playing time, and draft selections depend on accurate forecasts across a range of skills.  Multi-task models have been shown to have improved predictive power compared to one-at-a-time predictive models \citep{NIPS2007_66368270}, but to date such approaches have not been applied to forecasting basketball performance.

\end{enumerate}

Answering these questions requires a novel production curve model which must (i) treat the metrics jointly rather than independently (ii) place a prior over trajectories flexible enough to accommodate diverse career shapes yet constrained to a single peak, with the location of the peak inferred rather than fixed a priori \citep[cf.][]{vaci2019large}; (iii) respect the native sampling distributions of count and proportion metrics (iv) learn player similarity from the data rather than imposing it a priori through positional labels \citep{page2007using} or a pre-specified archetype sets \citep{vinue2019forecasting}; and (v) account for survivorship bias which arises due to the earlier retirement of less productive athletes.

To this end, we propose a Bayesian multi-task Concave Process Latent Variable Model (CPLVM) that jointly estimates concave production curves across over a dozen performance metrics and thousands of NBA players across 30 seasons of data.  The model accounts for player and metric similarity by learning player-specific latent embeddings from the full metric profile; players with similar embeddings have similar aging trajectories without relying on positional labels or hand-specified archetypes.  Concavity is enforced through a novel nonparametric prior which we call a ``concave process prior''.  A survival model accounts for right-censored careers, correcting the survivor bias that afflicts conditioning on observed seasons alone.

We use the results of our latent variable representation to characterize how players age across the suite of metrics. In doing so, we are able to generate posterior distributions of career trajectories across a wide set of metrics for prospective athletes with limited data. Furthermore, we generate a deeper understanding of when different metrics peak, characterizing which components of performance crest early and which crest late in a career. Finally, we are able to examine how a player moves through their aging profile, creating a characterization of athletes who take longer to develop vs. those whose decline is rapid.

The remainder of the paper is organized as follows.  In Section~\ref{sec:data} we describe the data and highlight the structure which motivates our modeling choices.  In Section~\ref{sec:proposed_model} we develop the proposed multi-task Bayesian model of performance, and in Section~\ref{sec:inference} we introduce the details of our inferential approach.  In Section~\ref{sec:results} we report the results of our model fit, including comparisons with alternative models and results which characterize variability in production curves across players and metrics.

\section{Data and Exploratory Analysis}

\label{sec:data}
In this work, we analyze data collected on approximately 2300 NBA players from years 1997 - 2026, from the ages of 18 - 38
scraped from \url{basketball-reference.com} at a yearly season-level resolution. The data can be naturally viewed as a  tensor of size $P$ by $T$ by $K$ where $P$ is the number of players, $T$ is the range of ages observed across players' careers, and $K$ is the number of production metric curves.   Values may be missing for a player at certain ages due to injury, retirement (right censoring), or late entry into the league (left-censoring). On offense, we measure free throw attempts (FTA), free throw percentage (FT\%), two point shot attempts (FG2A) and two point shot percentage (FG2\%), three point shot attempts (FG3A) and percentage (FG3\%), offensive rebounds (OREB), turnovers (TOV), and assists (AST).  On defense, we measure blocked shots (BLK), defensive rebounds (DREB), and steals (STL). 
We also include composite metrics like usage rate (USG) and box plus-minus (BPM), the latter of which is an estimate of points contributed by the player above league average per 100 possessions played \citep{myers2020bpm2}.  Box Plus Minus is divided into a measure of points added on offense (OBPM) and points added on defense (DBPM). In addition, we measure the average minutes per game played (MPG) and total fraction of games played in the season (GP\%). More detailed description of these metrics and an introduction to basketball analytics can be found in \cite{terner2020modeling}.

Broadly, there are three sources of variation, which align with each of the modes of the data tensor: variability over time for a given player-metric, covariability between players, and covariability between metrics.  In Figure \ref{fig:obpm} we plot LOESS smoothed curves of the evolution of OBPM over time for five different players.  We can see that for all five players they exhibit the expected pattern of improvement and decline.  However, the value at peak performance, the age at which they peak, and the rate at which they improve and decline all vary by player.  LeBron James, widely regarded as one of the best players of all time, exhibits the largest peak value OBPM of these five athletes and the slowest decline post-peak. Steve Nash, a famous "late bloomer", does not peak until his early 30s and sustains that ability until nearly 35 before steeply declining. We seek to model this temporal evolution of performance via a model that flexibly accounts for this variation across players. 

However, the rise-then-fall pattern is not limited to OBPM.  In Figure \ref{fig:with_cor} we show the correlation matrix of metrics over time, averaged over players with at least three seasons of data. The figure shows modest positive correlation between most metrics, which is driven by a combination of what the player can do \emph{athletically} and what \emph{skills} they have learned with experience (often referred to as "basketball IQ") or repetition.  For example, free throw attempt rate is correlated with two-point field goal attempt rate, since players who attack the rim are more likely to get fouled; the ability to get to the rim is largely driven by player athleticism.  Similarly, metrics measuring shooting skill tend to correlate, and players who rebound well on the offensive glass in a given season tend to rebound well on the defensive end. With the metrics ordered via hierarchical clustering, the heatmap shows a group of metrics related to shooting-skill and offensive performance, and a distinct group corresponding to rebounding and blocked shot metrics.  Composite metrics such as DBPM, which aggregate contributions of both kinds, correlate modestly with each group. A multi-task model of production should necessarily account for this metric-dependence.

Finally, we note that performance varies by player, depending on their playing style and position.  In Figure~S1 of the Supplementary Material we show the correlation in career average metrics computed between players. This plot largely reflects differences in playing style, best understood in terms by considering the three broad categories of basketball player: guard, forward, and center. Guards tend to be better at shooting and have more three point shots and centers tend to accrue more blocks and rebounds. This is apparent in the negative correlation between the group of blocks/rebounding metrics and the shooting metrics, since being a specialist in one set of skills tends to imply that they are worse at the other.  Positional categories should be considered useful heuristics, as playing style and ability exists on a spectrum \citep{chessa2023complex}. A production curve model should leverage the low dimensional structure inherent to player style and ability.

\begin{figure}[tbp]
    \centering
    \begin{subfigure}[b]{0.48\textwidth}
        \centering
        \includegraphics[width=\textwidth]{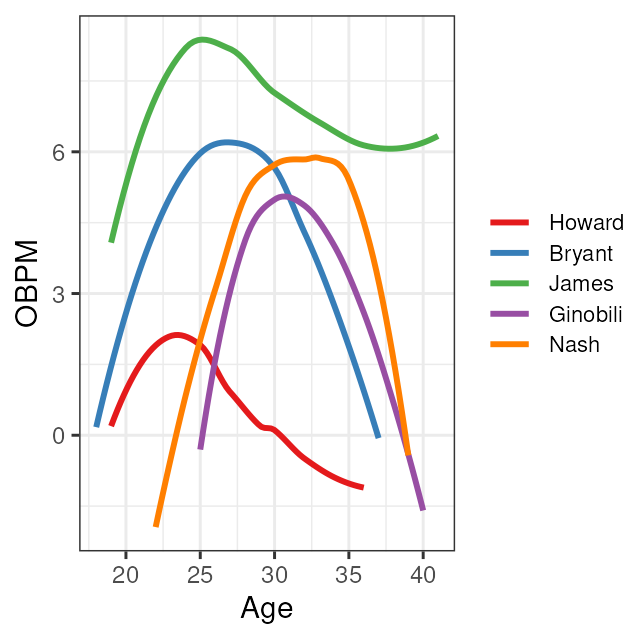}
        \caption{LOESS smoothing of OBPM over time}
        \label{fig:obpm}
    \end{subfigure}\hfill
    \begin{subfigure}[b]{0.48\textwidth}
        \centering
        \includegraphics[width=\textwidth]{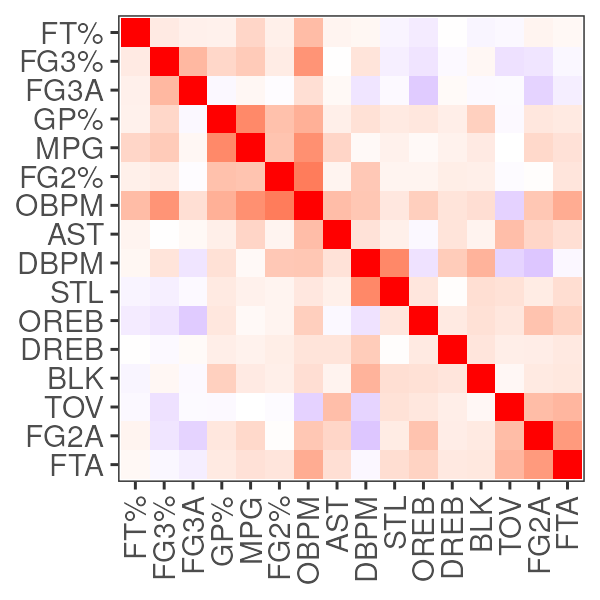}
        \caption{Within-player metric correlation}
        \label{fig:with_cor}
    \end{subfigure}

    \caption{(a) LOESS-smoothed OBPM trajectories for five players: every career rises and then declines, but peak age, peak value, and the rate of decline all differ by player. (b) Within-player correlation of metrics across seasons of the same career: most metrics are positively correlated over the course of a career. Rows and columns are ordered by hierarchical clustering.}
    \label{fig:exploratory}
\end{figure}

\section{Proposed Model}

\label{sec:proposed_model}

Let $\mathcal{Y}$ be the $P \times T \times K$ tensor of metrics $y_{ptk}$ be the measured outcome for player $p$ at age $t$, in metric $k$.  We model data at a yearly resolution between the ages of 18 and 38 ($t \in \{18, ..., 38\})$ over the 17 metrics described in Section \ref{sec:data} ($k \in \{1, ..., 17\})$. 
$y_{ptk}$ is missing for all $k$ if player $p$ does not play age $t$, either due to injury or retirement.

To account for the evolution of performance over time, we construct an aging function $f_k(X_p,t)$ that is concave in $t$ for all $p$ and $k$. Although this assumption than unimodality, we find that concavity provides a reasonable approximation to unimodality on a fixed domain. This aging function enters additively into a latent natural parameter $\mu_{ptk}(X_p)$, which represents the expected capability of player $p$ at age $t$ in metric $k$ given a player-specific latent embedding $X_p$. This parameter functions as a link-transformed conditional expectation across a number of distinct likelihood families, allowing us to model non-Gaussian data. We extend this GLM-like construction to a career duration survival model, which also models the natural parameters as a function of $X_p$.  In short, our model can be succinctly characterized via the following equations, 
\begin{align}
\mu_{ptk} &= f_k(X_p,t) + r^k_p(t) + o_{ptk} \label{eqn:mu_ptk} \\ 
y_{p t k} &\sim \mathcal{F}_k\left(g_k^{-1}\left(\mu_{p t k}\right), \phi_k, \omega_{p t k}\right) \label{eqn:link_fn}
\end{align}
where $f_k(X_p, .)$ represents a concave aging profile given player embedding $X_p$ for metric $k$ and captures aspects of performance that are predictable given the performance of similar players over time. $r_p^k(t)$ reflects aspects of the player's true skill that deviate from concavity in ways that are not predictable (e.g.\ due to injury, change of strategy or role).  This source of variability is given an AR(1) prior and is distinct from irreducible sampling variability due to finite observations which comes from the sampling distribution in \eqref{eqn:link_fn}. $o_{ptk}$ is a fixed precomputed offset that we use to adjust for seasonal league-average trends (see Appendix~B.1 of the Supplementary Material for details). $\mathcal{F}_k$ denotes the conditional distribution family corresponding to metric $k$ and  $g_k(\cdot)$ is a metric-specific link function mapping the conditional expectation to the linear predictor space.  $\phi_k$ represents optional metric-level dispersion or variance scale parameters and $\omega_{ptk}$ represents time-varying, observation-specific exposure weights or trial totals, such as minutes played or total field goal attempts.

In Subsection \ref{sec:likelihoods}, we describe the likelihood functions for each metric. In Subsection \ref{sec:player_similarity} we describe how we model similarity in structural variations from a shared low-dimensional latent space across players, which jointly maps each likelihood to a unified representation of career progression. Finally, 
in Section \ref{sec:concave_prior} we describe how we construct the concave process for modeling changes in production over time.

\subsection{Likelihood}
\label{sec:likelihoods}
In this section, we describe the data models, which include a career trajectory model across multiple metrics as well as a career duration survival model. Joint inference over the performance and survival components proceeds by summing the log-densities from both models, so that the latent embedding $X_p$ is simultaneously informed by what a player did on the court and how long they stayed in the league.

We consider four classes of observation models based on the available metrics: Gaussian, Binomial / Beta-Binomial, Poisson /  Negative-Binomial, and Beta (See Table \ref{tab:glm_framework}). Metrics which are inherently fractions (e.g.\ field goal percentage) are assumed to be conditionally Beta-Binomially or Binomially distributed, where $n_{ptk}$ is the number of trials, $\mu_{ptk}$ is the logit success probability, and $\mathcal{B}$ is the set of indices corresponding to made-shot metrics.  Raw volume counts (BLK, STL, AST, DREB, OREB, TOV) are modeled as conditionally Poisson,
where $\mathcal{C}$ is the set of Poisson indices.
\begin{table}[tbp]
\centering
\small
\setlength{\tabcolsep}{5pt}
\caption{Observation families by metric type: likelihood, link function, dispersion parameter, and exposure weight entering Equation \eqref{eqn:link_fn}. For notational brevity in the Binomial and Poisson case, we assume that the dispersion term from their overdispersed counterparts $\rightarrow \infty$.}
\label{tab:glm_framework}
\vspace{2mm}
\begin{tabular}{lllll}
\toprule
\textbf{Metric Type} ($k$) & \textbf{Likelihood} ($\mathcal{F}_k$) & \textbf{Link} ($g_k$) & \textbf{Dispersion} ($\phi_k$) & \textbf{Exposure} ($\omega_{ptk}$) \\
\midrule
\textbf{Made Shots} ($k \in \mathcal{B}$) & (Beta-)Binomial & Logit & Precision $\tau_k^2$  & Total attempts $n_{ptk}$ \\
\addlinespace
\textbf{Attempts} ($k \in \mathcal{C}, \mathcal{R}$) & Pois. / NB & Log & Concentration $\tau_k^2$ & Minutes played $m_{pt}$ \\
\addlinespace
\addlinespace
\textbf{Composite} ($k \in \mathcal{G}$) & Normal & Identity & Variance $\sigma^2$ & Precision weight $m_{pt}$ \\
\addlinespace
\textbf{Proportions} ($k \in \mathcal{P}$) & Beta & Logit & Precision $\tau_k^2$ & Mean-variance map \\
\bottomrule
\end{tabular}
\end{table}

\begin{equation}
y_{ptk} \sim \text{Binomial}(n_{ptk}, \operatorname{logit}^{-1}(\mu_{ptk}))  \text{ if } \text{k} \in \mathcal{B}
\end{equation}

\begin{equation}
y_{ptk} \sim \text{Poisson}(m_{pt} \cdot \exp(\mu_{ptk}))  \text{ if } \text{k} \in \mathcal{C}
\end{equation}

 Attempt totals (FTA, FG2A, FG3A), which exhibit overdispersion beyond Poisson, are modeled as Negative-Binomial,
\begin{equation}
y_{ptk} \sim \text{Negative-Binomial}(m_{pt} \cdot \exp(\mu_{ptk}), \tau^2_k)  \text{ if } \text{k} \in \mathcal{R}
\end{equation}
where $\tau^2_k$ is the overdispersion parameter and $\mathcal{R}$ is the set of Negative-Binomial attempt metrics. For continuous valued metrics (composite metrics such as OBPM and DBPM), we assume a conditionally Gaussian model:
\begin{equation}
y_{ptk} \sim \text{Normal}(\mu_{ptk}, \sigma^2 / m_{pt})  \text{ if } \text{k} \in \mathcal{G}
\end{equation}
where $\mathcal{G}$ corresponds to the indices of real-valued metrics. Proportion metrics bounded in $(0,1)$ — specifically usage rate (USG) and percentage of available minutes played (PCT Minutes) — are modeled as conditionally Beta distributed,
\begin{equation}
y_{ptk} \sim \text{Beta}(\operatorname{logit}^{-1}(\mu_{ptk}),\, \tau^2_k)  \text{ if } \text{k} \in \mathcal{P}
\end{equation}
where $\mathcal{P}$ is the set of indices corresponding to proportion metrics.

In addition, we consider some special cases related to playing-time.  Let $g_{pt}$ denote the total number of games played and $m_{pt}$ denote the minutes per game played by player $p$ at time $t$.  We assume that
\begin{align}
g_{pt} &\sim \text{Beta-Bin}(82, \operatorname{logit}^{-1}(\mu^g_{pt}), \tau_1^2)\\
m_{pt} &\sim 48\cdot \text{Beta}(\operatorname{logit}^{-1}(\mu^m_{pt}), \tau_2^2 /g_{pt})
\end{align}
That is, games played follows a Beta-Binomial distribution with 82 trials (the maximum games in a season) and minutes-per-game follows a scaled Beta distribution (scaled by 48, the maximum minutes in a game).  In both of the above, we use the mean-variance parameterization.

Finally, not every player is observed for the full span of seasons covered by the data, and their dropout is likely informative.  For instance, some players retire voluntarily, whereas others are forced out due to diminishing performance or injury; career lengths are heavily right-skewed, with the modal career lasting a single season (Figure~S2 of the Supplementary Material).  Ignoring these selection would bias estimates of the aging curve since players who exit early are systematically different from those who persist.  We therefore augment the trajectory likelihood with an explicit model for career duration that accounts for right-censored observations.

We model the  \emph{exit age} (e.g. retirement) of player $p$ using a survival model.  Let $T_p$ be the exit age and let $\alpha_p$ denote their \emph{entrance age}.  We place a player-specific \emph{Gompertz} proportional-hazards model on the time at risk $T_p \mid T_p > \alpha_p$, with age-varying hazard
\begin{equation}
    h_p(a \mid X_p) \;=\; \eta_p(X_p)\,\exp\!\bigl(\gamma_p(X_p)\,a\bigr),
    \label{eq:gompertz}
\end{equation}
where $\eta_p(X_p) > 0$ is the \emph{baseline hazard} and $\gamma_p(X_p) > 0$ is the \emph{aging rate} governing how quickly exit risk accelerates with age.  
We adopt the Gompertz form because its log-hazard is linear in age, $\log h_p(a) = \log\eta_p(X_p) + \gamma_p(X_p)\,a$, encoding the notion that exit risk increases by a constant proportional factor with each additional year of age.  This exponential-in-age hazard is commonly used to model human aging, making it a natural model of athletic decline. More details of this formulation can be found in Appendix~B.4 of the Supplementary Material.

\subsection{Player Specific Latent Variables}
\label{sec:player_similarity}

A key goal is to accurately model dependence between both players and metrics in a data-informed manner.  To this end, we borrow heavily from the methodology of Gaussian Process Latent Variable Models (GPLVMs), which treat the usual input to the Gaussian Process as an unobserved latent variable \citep{JMLR:v6:lawrence05a}. Specifically, we introduce a $q$-dimensional player-specific latent variable $X_p \in \mathbb{R}^q$, such that players that with similar values of $x_p$ are expected to have more similar career production curves across metrics. We encode prior domain knowledge about player embeddings using a \emph{structured} prior mean, by shrinking latent variables towards a linear combination of exogenous player attributes:
\begin{equation}
\label{eq:structured_prior}
  X_p = Z_p^\top W + \sigma_X\, \varepsilon_p, \quad
  \varepsilon_p \sim t_4\!\bigl(0,\, \sqrt{1/2}\; I_q\bigr), \quad
  W \sim \mathcal{N}(0,\, \sigma_W^2\, I)
\end{equation}
Covariates, $Z$, include height, draft position, and nominal position (Center, Guard, Forward). The residual is heavy-tailed, scaled to unit marginal variance so that it matches a standard normal in width and differs only in the tails, which lets an outlying player sit away from the population without widening the bulk of the embedding. A more detailed description of this prior can be found in Appendix~B.2 of the Supplementary Material. This prior helps identify more appropriate latent representations for players with limited data, which in turn improves predictive performance.  Next, we describe how these variables factor into the expected production for players across all metrics.

\subsection{A Non-Parametric Prior for Concave Functions}
\label{sec:concave_prior}
 In this Section we describe our construction for the concave functions $f_k(X_p, t)$.  Specifically, we construct $f$ by developing a novel nonparametric prior over concave functions which strikes an ideal balance between flexibility over simpler parametric alternatives (e.g. quadratic models) while appropriately constraining the functions to respect expected well-documented aging patterns. We first describe this prior in general terms  before describing how it's applied specifically to model basketball data.

We call the prior the ``concave process prior'' which arises as a composition of continuous push-forward maps applied to a Gaussian Process with a stationary kernel. Our approach is inspired the strategy used by \cite{convexityconstraints}, who propose a non-parametric prior over the space of monotonic functions.
\begin{proposition}[Concave process prior]
We say a random function f(t) has a concave process (CP) prior iff
\begin{align}
    f(t) = \beta_0 + \beta_1 (t-t_0) - \int_{t_0}^t \int_{t_0}^s [g(z)]^2 \, dz \, ds,
    \label{eqn:concave_process_uni}
\end{align}
where $g \sim \mathcal{GP}(0, K_t)$ is a mean-zero Gaussian process with stationary covariance kernel $K_t(\cdot, \cdot; \Theta)$, and $\beta_0, \beta_1$ are random variables with prior support over all $\mathbb{R}$.
If $f(t) \sim CP$, it is almost surely concave.  Moreover, let $\mu$ denote the law of $g$ on $C([t_0, t_1])$, the space of continuous functions on the modeled interval, and let $T : C([t_0, t_1]) \to C([t_0, t_1])$ denote the map defined by the right-hand side of Equation~\eqref{eqn:concave_process_uni},
\begin{align*}
T(g) &= \beta_0 + \beta_1 (t-t_0) - \int_{t_0}^t \int_{t_0}^s g(z)^2\, dz\, ds.
\end{align*}
Then $f(t)$ puts positive mass on every uniform neighborhood of any continuous concave function.
\end{proposition}

\begin{proof}
See Appendix~A.1 of the Supplementary Material
\end{proof}

The construction above requires only that $K_t$ be stationary with a well-defined spectral density; throughout the remainder of the paper we take $K_t$ to be the squared-exponential kernel.  In practice, the double integral in Equation \eqref{eqn:concave_process_uni} cannot be evaluated analytically, so we make a linear approximation to the Gaussian process as in \citet{convexityconstraints}.  Specifically, we use the Hilbert Space GP basis function approximation (HSGP), which converges as the number of basis functions goes to infinity, and leads to tractable analytic solutions \citep{Solin_2019}. See Appendix~C of the Supplementary Material for details.

\begin{proposition}\label{prop:concave_uni_gp}

In Equation \eqref{eqn:concave_process_uni} replace $g(z)$ with the HSGP approximation $\sum_{l=1}^J S(\sqrt{\lambda_l})^{1/2} \alpha_l \psi_l(t)$, where $S(\cdot)$ is the spectral density of the kernel $K_t$, $\alpha_l \overset{iid}{\sim} \mathcal{N}(0,1)$ are basis weights, $\psi_l(t)$ and $\lambda_l$ are the $J$ eigenfunctions and eigenvalues of the HSGP approximation of $K_t$ on the compact interval $[-L, L]$, If we let $t_0 = -L$, then, 
\begin{align}
\label{eq:final_uni_gp}
    f(t | \Theta) = \beta_0 + \beta_1(t+L) - \mathbf{\alpha}^T \Psi(t) \mathbf{\alpha}
\end{align} 
where 
\begin{align}
\label{eq:phi}
\Psi_{lz}(t) = 
\begin{cases} \frac{cos(\epsilon_{lz}^{+}(t+L)) - 1}{2L(\epsilon_{lz}^{+})^2} +
\frac{(t+L)^2}{4L}  & \text{if } l = z \\
\frac{cos(\epsilon_{lz}^{+}(t+L)) - 1}{2L(\epsilon_{lz}^{+})^2} - \frac{cos(\epsilon_{lz}^{-}(t+L)) - 1}{2L(\epsilon_{lz}^{-})^2} & \text{if } l \neq z
\end{cases}
\end{align}
Here, $\mathbf{\alpha}$ is a $J$-vector with entries given by $\alpha_l \cdot S(\sqrt{\lambda_l})^{1/2}$, and $\epsilon_{lz}^{+} = \sqrt{\lambda_l} + \sqrt{\lambda_z}$, $\epsilon_{lz}^{-} = \sqrt{\lambda_l} - \sqrt{\lambda_z}$, respectively. $\Theta$ collects the hyperparameters of $K_t$ (lengthscale and amplitude), which enter through the spectral density $S$.  The parameters $\beta_0$ and $\beta_1$ are the constants of integration carried over from Equation~\eqref{eqn:concave_process_uni}.
\begin{proof}
    see Appendix~A.2 of the Supplementary Material.
\end{proof}
\end{proposition}

The the analytic approximation to concave process derived in Equation \eqref{eq:final_uni_gp} can thus be incorporated into a fully Bayesian model that can be implemented in probabilistic programming language.  We further reparameterize this model into a form that more amenable to prior elicitation for NBA production curve models.  Specifically, note that $\beta_0$ and $\beta_1$ can be modeled as a-priori independent parameters that control the shape of the concave function. For modeling NBA production, we are most interesting in reasoning about both the value and age of peak performance. As such, define $f^*$ and $t^*$ the max and argmax of $f(t)$ respectively, that is

\begin{align}
f^* = \underset{t}{\text{max }} f(t) = f(t^*) 
\end{align}
In the following Theorem, we state the parameterization of the concave process in terms of $f^*$ and $t^*$.
\begin{theorem}\label{theorem:concave_max_approx}
Let $f(t | \Theta)$
be defined as in Equation \eqref{eq:final_uni_gp}. Then, if $f(t | \Theta)$ takes on the maximum value of $f^*$ at time $t^*$,  

\begin{align}
\label{eq:max_constraint_gp}
    f(t | \Theta) = f^* + \mathbf{A}^T [ \Psi(t^*) - \Psi(t) + \frac{d\Psi(t^*)}{dt}(t - t^*) ] \mathbf{A}  
\end{align}
where $\frac{d\Psi(t^*)}{dt}$, the time derivative of \eqref{eq:phi}, can be computed analytically.
\end{theorem}
\begin{proof}
See Appendix~A.4 of the Supplementary Material.
\end{proof}

Because all of the relevant link functions are strictly monotone increasing, the location of the maximum is preserved under the inverse link: $\arg\max_t \, g^{-1}(f(t)) = \arg\max_t \, f(t)$.  Consequently, the peak age $t^*_k$ and the rank ordering of peak values are exactly preserved in the observable space, and the model's ability to place a structured prior on peak age and peak value is maintained regardless of the observation family.

\subsubsection{Modeling Dependence Across Concave Processes}

In this Section, we extend the univariate concave process to a process over multiple concave functions.  Here, we assume each process is associated with auxiliary time-invariant covariates, $x$, which in our application corresponds to the player-embedding.  We let $f(x,t)$  denote a multivariate concave process associated with covariate $x$ which is concave in $t$ and model the dependence between $h(x,t)$ and $h(x',t)$ for $x \neq x'$. We extend the concave process to a multivariate processes by using a multi-output Gaussian process.

 We can construct dependent multivariate concave functions, $f(x,t)$, by extending \eqref{eq:final_uni_gp} by integrating the square of a multi-task Gaussian Process, $g(x,t)$. Specifically, $g(x,t) \sim \mathcal{GP}(0, K_x \otimes K_t)$, where $K_x$ and $K_t$ are stationary covariance kernels acting on the $x$ space, and time $t$, independently, and by allowing the integration constants $\beta_0$ and $\beta_1$ to also depend on $x$.
We wish to find all set of $f(x,t)$ such that $[g(x,t)] ^2 = \frac{\partial ^2 f(x,t)}{\partial t^2} \geq 0$. As such, our desired function has the form
\begin{align}
\label{eq:concave_multi_rep}
    h(x,t) = \beta_0(x) + \beta_1(x)(t-t_0) - \int_{t_0}^t \int_{t_0}^s [g(x,z)]^2dzds 
\end{align}
where $\beta_0(x), \beta_1(x)$ are constants of integration with respect to $t$.

\begingroup\sloppy
\begin{proposition}\label{prop:concave_multi_gp}
Assume \eqref{eq:concave_multi_rep} holds and let $g(x,t) = \sum_{i=1}^M \phi_i(x) \bigl( \sum_{l=1}^K S(\sqrt{\lambda_l})^{1/2} \alpha_l^i \psi_l(t)\bigr)$, where $\phi_i(x)$ is a projection of $x$ such that $\phi(x)^T \phi(x) \approx \mathcal{K}(x,x')$ and $\psi_l(t)$ are HSGP basis functions for $t$. Then,
\begin{align}
\label{eq:final_gp}
    f(x,t| \Theta) &= f^{*}(x) \nonumber\\
    &+  \mathbf{\phi}(x)^T\mathbf{A}^T \!\left[ \Psi(t^{*}(x)) - \Psi(t) + \tfrac{d\Psi(t^*(x))}{dt} (t - t^{*}(x)) \right] \mathbf{A}\, \mathbf{\phi}(x)
\end{align}
where $\Psi_{lz}(t)$ is defined as in Equation \ref{eq:phi}, and $\mathbf{A}$ is a $M$ by $K$ matrix where $\mathbf{A}_{il} = S(\sqrt{\lambda_l})^{1/2} \alpha_l^i$. Our formulation $f(x, t | \Theta)$ converges uniformly to $h(x,t)$ of Equation \eqref{eq:concave_multi_rep}
as $K, M \xrightarrow{} \infty$.

\noindent Proof: see Appendix~A.3 of the Supplementary Material.
\end{proposition}
\endgroup

Much like the univariate case in Theorem~\ref{theorem:concave_max_approx}, we place independent
Gaussian process priors on the peak value $f^{*}$ and the peak age $t^{*}$, both indexed by the
player embedding, so that players with similar latent coordinates are given similar peaks:
\begin{align}
  f^{*}_k(X_p) &\sim \mathcal{GP}\!\left(c_k,\; \sigma_c^2\, \mathcal{K}(\cdot,\cdot)\right), \nonumber\\
  t^{*}_k(X_p) &= a \tanh\!\left( g_k(X_p) + \tanh^{-1}\!\left(t_k / a\right) \right), \qquad
  g_k \sim \mathcal{GP}\!\left(0,\; \sigma_t^2\, \mathcal{K}(\cdot,\cdot)\right),
  \label{eq:peak_priors}
\end{align}
where $\mathcal{K}$ is the same latent-space kernel that governs the curvature weights, and $c_k$
and $t_k$ are metric-specific offsets.  The peak value is therefore a Gaussian process on the
embedding, and the peak age is the same construction bounded such that the peak age falls within our observation window.
Each metric is realized as an independent draw from these processes, while the shared kernel $\mathcal{K}$ ties
them to a common notion of player similarity.  Both processes, like the curvature weights, are realized in practice through
the random-feature approximation described in Appendix~B.5 of the Supplementary Material.

Equations~\eqref{eqn:mu_ptk}, \eqref{eq:final_gp} and \eqref{eq:peak_priors} together specify the production model.  For player $p$ and metric $k$, the aging curve $f_k(X_p, t)$ is the concave process of Equation~\eqref{eq:final_gp} evaluated at $x = X_p$, its peak located and scaled by the two processes of Equation~\eqref{eq:peak_priors}, and its curvature determined by the metric-specific weight matrix $\mathbf{A}_k$, which couples the latent space to the $M$ HSGP time bases.  We give $\mathbf{A}_k$ independent normal entries scaled by the square root of the spectral density $\sqrt{S(\sqrt{\lambda_l})}$, so that the induced prior on the second derivative is a Gaussian process over age with per-metric amplitude $\alpha_k$ and lengthscale $\ell_k$, indexed by the same embedding as the peak processes.  The remaining term in Equation~\eqref{eqn:mu_ptk}, the idiosyncratic deviation $r^k_p(t)$, is given an AR(1) prior with autocorrelation $\rho_k \sim \mathrm{Uniform}(-0.5, 0.5)$ and innovation scale $\sigma^k_{\mathrm{AR}} \sim \mathrm{HalfNormal}(0.3)$.  The full set of priors, including the dispersion and exposure parameters of each likelihood family in Table~\ref{tab:glm_framework}, is given in Appendix~B.6 of the Supplementary Material.

\section{Inference and Hyperparameter Selection}
\label{sec:inference}

We implement our model in the probabilistic programming language \texttt{NumPyro} \citep{numpyro}. Full MCMC inference in this model is challenging Bayesian factor models typically have rotational identifiability issues with either the loadings or the factors themselves, yielding a multi-modal posterior. In particular, we find that MCMC sampler gets stuck at suboptimal modes when randomly initialized.  As such, to start the the sampler near a well-identified mode by initializing the MCMC chains from a maximum a posteriori (MAP) estimate. The MAP estimate is used to fix a small set of kernel hyperparameters, specifically, the concave GP lengthscale and amplitude ($\alpha$, $l_k$, $\sigma_t$, $\sigma_c$), the structured prior scales ($\sigma_W$, $\sigma_X$), and the Gompertz global hazard offset parameters. All remaining parameters  are sampled. 

The model was implemented in the probabilistic programming language \texttt{NumPyro} \citep{numpyro} and inference run on an NVIDIA RTX A5500 GPU (24\,GB).  The MAP estimate was obtained by Stochastic Gradient Descent (SGD) and MCMC samples acquired using No U-Turn Sampling (NUTS). Posterior samples were obtained from 4 chains, run in parallel across both GPUs after 2{,}000 warmup iterations, with every post-warmup draw retained.  Convergence diagnostics are given in Appendix~E and Table~S2 of the Supplementary Material.  
We select the dimension of the latent player embedding, $q$, by conducting a holdout ablation study (Appendix~H of the Supplementary Material).  A latent dimension of $q = 10$ was used for the player embedding space, along with 5 HSGP basis functions, to approximate the CPLVM. Analysis of the resultant rotationally non-identifiable latent space $X$ is conducted through a Procrustes alignment (Appendix~B.3 of the Supplementary Material). 

Finally, we describe our procedure for generating posterior predictive draws. This is non-trivial as many metrics are conditional on each other through the exposure structure: count-based metrics depend on minutes played, which in turn depends on games played all of which are part of the joint forecast in our CPLVM model. To account for this, we generate posterior predictive curves in systematic way.  First, we sample games played (GP\%) and minutes per game from the posterior predictive distribution.  The total-minutes in a season is generated as the product of these two quantities and serves as the exposure for for all counting and rate metrics (FTA, AST, DREB, etc.). Shooting percentages are then derived from binomial given the the sampled counts.  This cascade appropriately propagates uncertainty about playing time into the uncertainty for season total production.  All code for the analyses in this paper can be found in the \href{https://github.com/abrahme/nba\_functional\_analysis}{github repository}.

\section{Results}
\label{sec:results}
In this Section, we highlight a select number of ways in which our model can be used. First, in Section \ref{sec:validation} we demonstrate the superior out of sample forecasting  performance of our CPLVM model compared to many reasonable alternatives.  We show that multi-task inference, i.e. fitting all metrics jointly, improves predictive performance over one-at-a-time models. Then, we show that the concave process model outperforms both more flexible nonparametric models which are not constrained to concavity as well as less flexible parametric alternatives.  In Section \ref{sec:posterior_curves}, we demonstrate the model fit across a number of players and show how it leads to useful forecasts of future production as well as principled retrieval of historically comparable players.
Finally, in Section \ref{sec:uncovering_metrics}, we summarize aging patterns vary across metrics, focused on the timing of the peak and how sensitive different performance metrics are to age-related change.

\subsection{Model Comparison}
\label{sec:validation}

A key advantage of the joint model is its ability to borrow information across metrics and across similar players.
Even for an established star with many seasons of data, the model in Equation~\eqref{eqn:mu_ptk} fit to a single metric sees only that metric's own history.
The joint model instead exploits the low-dimensional structure in how metrics co-evolve: the trajectories of a player's other metrics are informative about how OBPM \emph{should} evolve, which should produce tighter and more coherent posteriors.

We test this pooling directly by  refitting the model to OBPM alone and find that doing so significantly degrades held-out predictive accuracy \emph{on OBPM itself} (Appendix~F and Table~S3 of the Supplementary Material).

We further compare our concave process approach to three alternative models.  First, we consider an unconstrained Gaussian process (GP) prior on $f_k(X_p, t)$, built from the same RFF basis, which imposes smoothness but no restriction on the shape of the trajectory.  We also compare the CP prior to parametric alternatives based on a simpler quadratic model.  We consider two variants: an \emph{asymmetric} quadratic, in which curvature is constant on each side of the peak but allowed to differ between the ascent and the decline, and a \emph{symmetric} quadratic, in which a single constant curvature governs both.
These models all incorporate player-specific latent embeddings, $X_p$, the random-Fourier-feature latent kernel, the per-player autoregressive term, and the position-year de-trend as described in Equation \eqref{eqn:mu_ptk}. They differ only in what kind of constraints we impose on $f_k(X_p, t)$ over time.  Specifically, these models can be considered nested in the sense that each 
\[
  \underbrace{\text{GP}}_{\frac{df^2}{dt^2}\ \text{free}}
  \;\supset\;
  \underbrace{\text{concave}}_{\frac{df^2}{dt^2} \leq 0}
  \;\supset\;
  \underbrace{\text{asym.\ quadratic}}_{\frac{df^2}{dt^2}\ \text{piecewise const.}}
  \;\supset\;
  \underbrace{\text{sym.\ quadratic}}_{\frac{df^2}{dt^2}\ \text{const.}}
\]
As such, the concave process prior used as the featured model represents a balance between nonparametric flexibility and informed constraints.  

Differences between these nested variants are quantified by the standard deviation of the \emph{paired} pointwise differences, $\mathrm{se}_{\Delta} = \sqrt{n\operatorname{var}_i(\mathrm{elpd}_i^A - \mathrm{elpd}_i^B)}$, following \citet{vehtari2017practical}.     Because the chain is nested, we compare each variant against its immediate parent rather than against whichever model scores best in a given scheme, so that every entry measures the cost of one specific restriction against a reference that is fixed across schemes. 

Constraining the unconstrained GP to a single peak improves held-out predictive density on every scheme. The parametric restrictions beyond it do not help, suggesting that the flexibility of a non-parametric model is necessary.

Table~\ref{tab:all_summary} reports holdout $95\%$ HDI coverage together with this ladder: its columns run from least to most restricted, and each ELPPD entry is the change in held-out density per observation from imposing that column's restriction on the model to its left, with $\dagger$ marking differences exceeding $1.96\,\mathrm{se}_{\Delta}$.  Coverage for the  model sits two to four points below the 95\% across schemes.  One possible source for this under-coverage include the fact that we fix a set of hyperparameters to their MAP estimate (Section~\ref{sec:inference}), which removes one layer of uncertainty from the predictive distribution.

\input{figs/model_plots/coverage/combined_all_summary.tex}

\subsection{Forecasting and player similarity}
\label{sec:draft_evaluation}
\label{sec:posterior_curves}

Forecasting player ability is a central goal for front offices, who are constantly looking to anticipate how a player's production will evolve before committing draft capital, roster spots, and multi-year salary to that projection.  Relatedly, identifying historically similar players helps contextualize a young player's early seasons: the completed careers of comparable predecessors provide a reference for how such profiles have historically developed and declined. As a case study, we focus on the recent breakout star, seven-foot-four-inch Victor Wembanyama.  Wembanyama entered the league in 2023 amid historic expectations, and his combination of size, rim protection, and perimeter skill has few direct precedents.

In Figure \ref{fig:wemby_prior_posterior}, we illustrate how our structured prior on the latent space adapts to new data.  Because Wembanyama was the number one overall pick, the structured prior on the player embedding (Equation~\ref{eq:structured_prior}) anchors $X_p$ in sparsely populated region occupied by elite big men such like Tim Duncan and Alonzo Mourning. After three seasons however, the data shift his posterior further toward the edge of the population, into the neighborhood of Durant and James, indicating that Wembanyama has more guard like qualities than is typical of a player of his size. In order to find comparable players, we rank players with completed careers by the distance between their posterior mean embedding and Wembanyama's, computed after Procrustes alignment of the posterior draws (Appendix~B.3 of the Supplementary Material). His five nearest neighbors are David Robinson, Alonzo Mourning, Patrick Ewing, Derrick Coleman and Hakeem Olajuwon.  Four of the five are rim-protecting centers; the fifth, Derrick Coleman, is a power forward.

In Figure~\ref{fig:wemby_neighbor_curves} we overlay the posterior mean production curves of Wembanyama and his five nearest completed-career neighbors across AST, BLK, MPG and OBPM, illustrating how closely the model places his projected career alongside theirs.  Figure~S3 of the Supplementary Material shows a complementary view: the model forecast for Wembanyama across OBPM, FG2\%, FG2A and FTA given data from three recorded seasons of play (2023--2026), shown above the fits to career data from David Robinson and Alonzo Mourning with observed data overlayed. From the model projections, we expect Wembanyama's offensive performance, as measured by OBPM, to peak at around 25 years of age. 

\begin{figure}[tbp]
  \centering
  \includegraphics[width=0.5\linewidth]{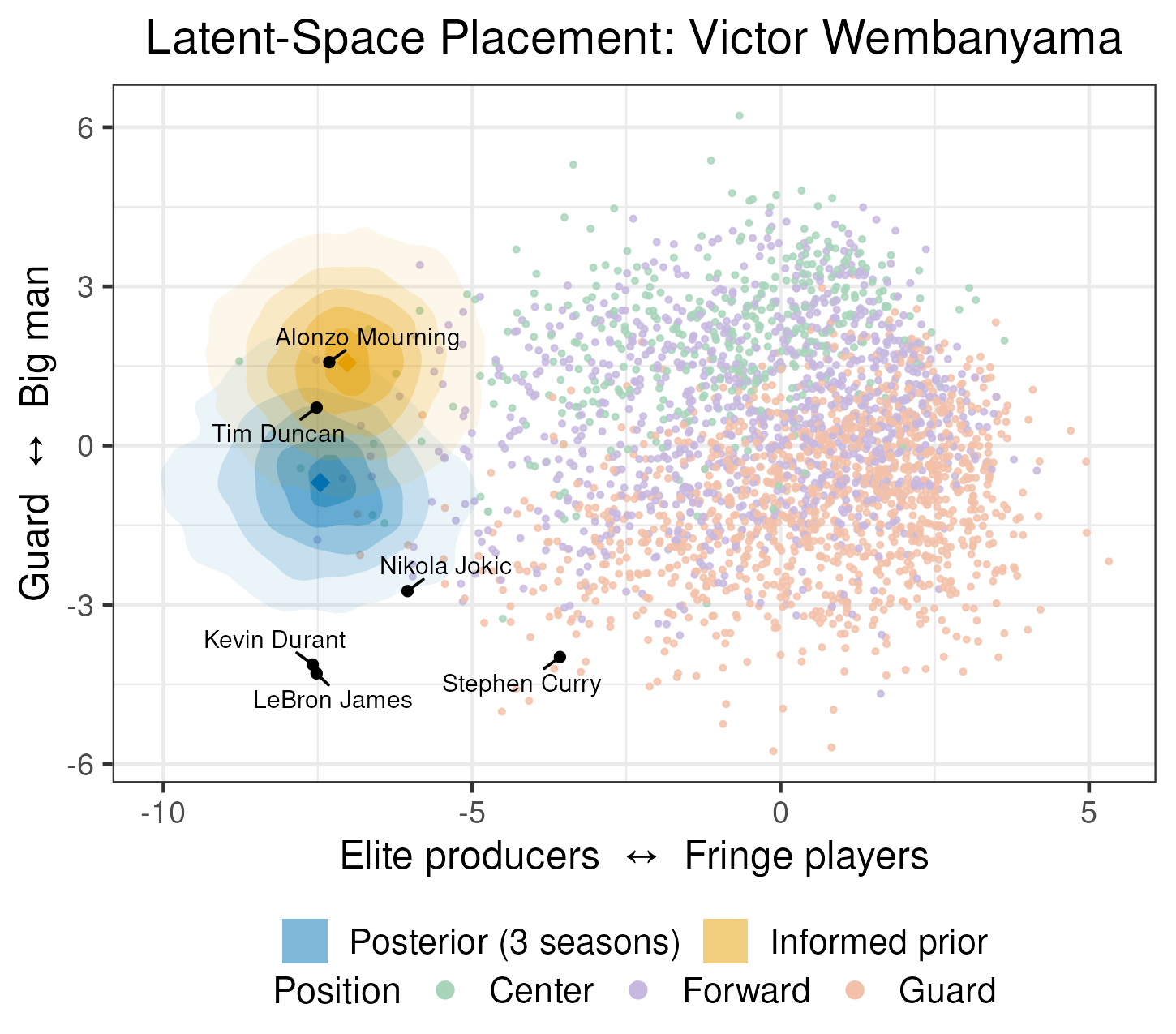}
  \caption{Prior and posterior placement of Victor Wembanyama in the latent space.  Background points show all players at their posterior mean embeddings, projected onto the first two principal components of that embedding and coloured by listed position. The horizontal axis separates elite producers from fringe players, and the vertical axis separates guards from big men.  The orange density reflects the covariate-anchored prior density of $X$ for Wembanyama (Equation~\ref{eq:structured_prior}). The blue density depicts the posterior over $X_p$ given three observed seasons of play. Diamonds mark the prior and posterior means.  Draft standing, height, and positional role alone place him at the thinly populated big-man edge of the population near Tim Duncan and Alonzo Mourning; the observed seasons move him further outward, toward Kevin Durant and LeBron James.}
  \label{fig:wemby_prior_posterior}
\end{figure}

\begin{figure}[tbp]
  \centering
  \includegraphics[width=\linewidth]{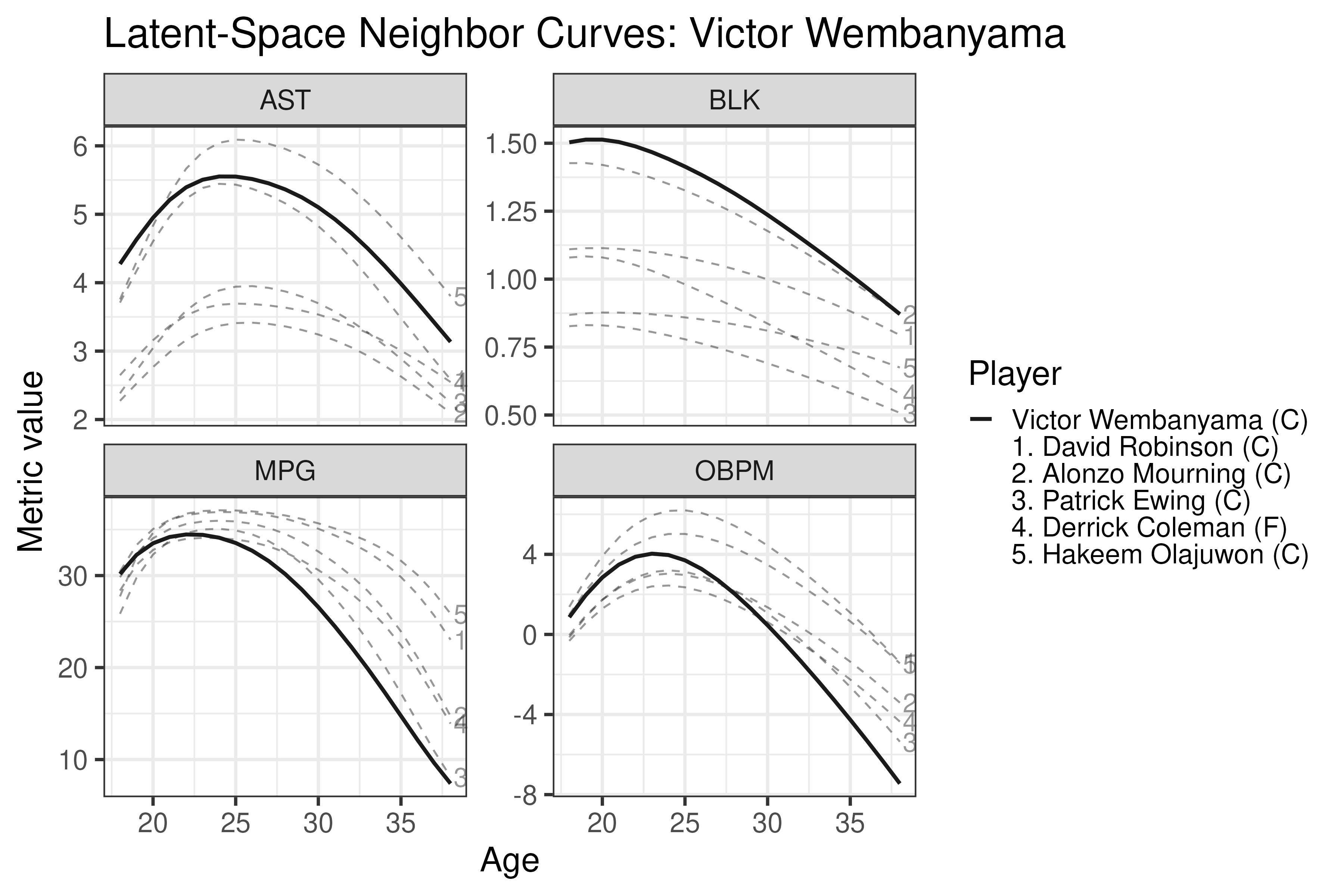}
  \caption{Posterior mean trajectories of $f^k(x_p, t)$ for Victor Wembanyama (solid line) and his five nearest latent-space neighbors, across AST, BLK, MPG, and OBPM. Neighbors are restricted to completed careers (at least ten observed seasons, with observations at age 34 or later).  Curves are on each metric's natural scale: counts per 36 minutes, MPG in minutes per game, and OBPM in points added per 100 possessions.}
  \label{fig:wemby_neighbor_curves}
\end{figure}

Next, we extend our forecasts to all athletes that entered the league between 2020 and 2025.  We illustrate the forecasting capability of our model, focusing on OBPM as a high-signal measure of overall offensive production. In Table~\ref{tab:breakout_cohort} we report the top five projected players per entry cohort ranked by the posterior mean of their peak OBPM. We also include 95\% credible intervals for their peak OBPM, the estimated age that they peak (and credible interval) and posterior probabilities that their career peak OBPM exceeds 2, 3 and 4. These thresholds span the range from above-average starter to All-Star caliber production: only 19 percent of NBA players in our dataset have a measured value of OBPM greater than $2$ by the end of their career.  This drops to 12\% for players that ever exceed an OBPM of $3$ and $7\%$ for players that ever exceed an OBPM of $4$.  The 2020--2021 entry cohorts show high mean peak OBPM for their top players (Zion Williamson, Tyrese Haliburton, Anthony Edwards) with near-certain $P(\text{OBPM}{>}2)$, while the 2025 cohort, with only two seasons of data, yields more diffuse probabilities. Victor Wembanyama tops the 2024 entry cohort, with an expected OBPM peak of 4.2. For reference, only about a third of number one overall picks like Wembanyama ever exceed an OBPM of 4.  Thus, while a priori we expect this outcome to be somewhat unlikely, after observing three years of play, the model is more confident that Wembanyama is likely to be an elite caliber player.

\input{figs/nba_convex_max_tvrflvm_AR_posyear_long/stratified_next_k/mcmc/plots/latent_space/map/breakout_by_cohort.tex}

\subsection{Characterizing Aging Patterns}
\label{sec:uncovering_metrics}
In this Section, we explore variation in the timing of peak performance both across metrics and players. In order to understand heterogeneity in the timing of peak performance across metrics, we compute the posterior mean of each player's age of peak performance across metrics. This is only made possible by the reparameterization of the concave process to explicitly encode $f^{*}$ and $t^{*}$, which allows us to explicitly sample the posterior of the peak value for a metric, as well as the peak age. These posterior values act as smoothed estimates of an individual's observed peak value and peak age, allowing us to perform Principal Component Analysis (PCA) on the resultant matrix without introducing noise from the observational data. We plot the loadings for each of the metrics for the first two principal components, as well as the principal component scores in Figure~\ref{fig:metric_loadings}. The first principal component captures variation in timing of the peak for nearly all metrics, separating early peakers (Jonny Flynn) from late (Steve Nash), while the second component contrasts the timing of rebounding, shot-blocking, defensive impact and two-point efficiency with that of three-point attempt volume.

\begin{figure}[tbp]
    \centering
    \begin{minipage}[t]{0.49\linewidth}
      \centering
      \includegraphics[width=\linewidth]{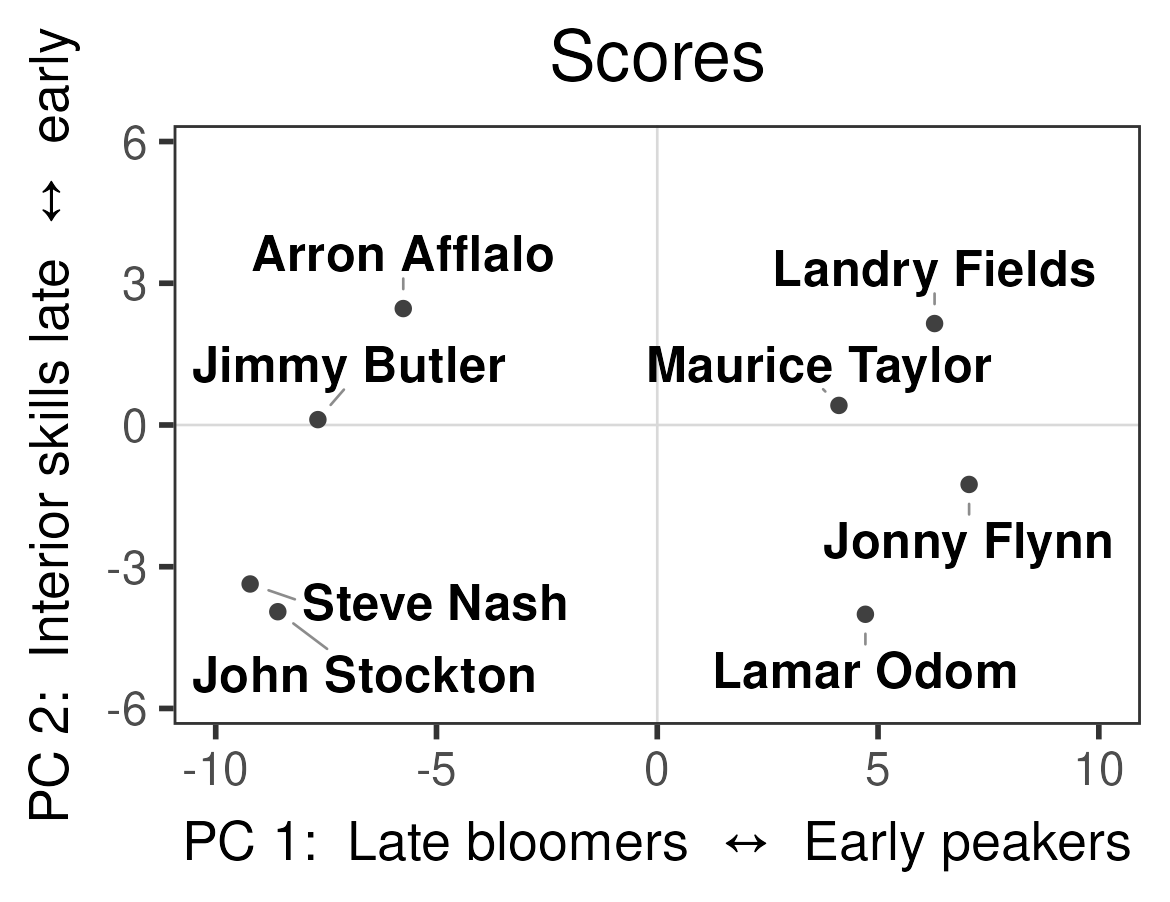}
    \end{minipage}%
    \hfill
    \begin{minipage}[t]{0.49\linewidth}
      \centering
      \includegraphics[width=\linewidth]{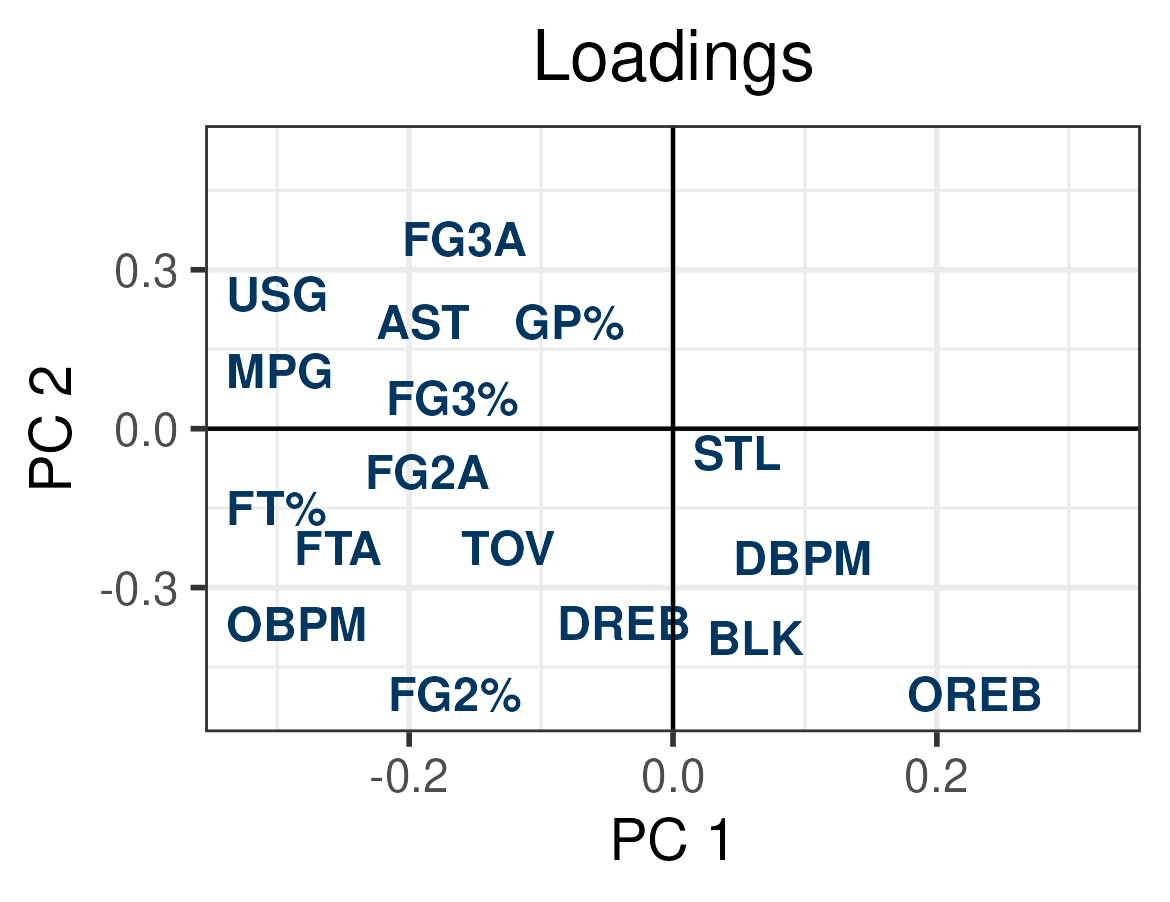}
    \end{minipage}

    \caption{PCA on posterior mean peak ages by metric: player scores on the first two principal components (left) and the corresponding metric loadings (right). In the left panel, two representative players are labelled per quadrant. The horizontal axis (PC1) captures overall peak timing. All but four metrics load with the same sign (the exceptions are DBPM, STL, BLK, and OREB), so players toward the left peak later across most of the board (e.g.\ Nash, Stockton) and those toward the right peak earlier (e.g.\ Jonny Flynn, Landry Fields). The vertical axis (PC2) contrasts the timing of rebounding, shot-blocking, defensive impact and two-point efficiency with that of three-point attempt volume. Players near the top have their three-point volume peak earlier (e.g.\ Arron Afflalo), while players near the bottom show the reverse (e.g.\ Nash, Stockton).}
    \label{fig:metric_loadings}
\end{figure}

In Figure~\ref{fig:peaks} we plot the population distribution of the age of peak performance, split by metric. The vast majority of metrics peak between the ages of 20 and 25, with the exception of assists, DBPM, and three point field goal attempts. We find that metrics that are  more dependent on athleticism tend to peak earlier.  For example, blocking a shot requires vertical explosiveness and the recovery speed to contest at the rim.  Offensive and defensive rebounding reward second-jump quickness and the ability to hold position in traffic.  By contrast, FT\%, FG3\%, and AST  are primarily skill-driven since they largely reflect the result of a technique refined through repetition \citep{ericsson1993deliberate}. These metrics tend to peak later. Three-point attempt rate is the clear outlier, continuing to rise into a player's early thirties as shooters add volume, perhaps to compensate for their declining ability to generate two point field goals. 

In order to properly contextualize these values, we also compute measures of how much change we can expect in a career, relative to between-player-variation in peak performance.  Specifically, we compute the expected yearly improvement over the four years before the peak, $\mathbb{E}_p[(f_p^{*} - f_p(a^{-}))/(t^{*} - a^{-})]/\mathrm{sd}_p(f_p^{*})$, as well as the expected yearly decline over the four years after it, $\mathbb{E}_p[(f_p^{*} - f_p(a^{+}))/(a^{+} - t^{*})]/\mathrm{sd}_p(f_p^{*})$, taken over all players for each specific metric, where $a^{-} = \max(t^{*}-4,\,18)$ and $a^{+} = \min(t^{*}+4,\,38)$ stop the window at the boundary of the modelled ages. These values are depicted for all metrics in Figure~\ref{fig:peaks_combined}.  Metrics further from the origin  change the most over a career, whereas those that are closer to the origin  tend to change relatively little. This suggests that aging greatly affects metrics such as minutes per game, OBPM, games played while leaving metrics such as blocks, defensive rebounds, and steals relatively intact. 

Blocks, for example, ``peak'' very early (even before most athlete's enter the NBA) but do not decline much over an athlete's career, suggesting block rate is dependent more so on player style rather than aging. 
On the other hand, OBPM, FG2\% and minutes per game change more dramatically over an athlete's career. If we assume that shooting in isolation is a technique driven ability, then we can attribute the aging derived change in two point field goals to a drop in athleticism which prevents the shooter from generating clean physical separation from their defenders. This in part can explain drops in OBPM and minutes per game, as a diminished ability to pressure defenses can cause a drop in impact on the offensive end and meaningful minutes. 

We also identify metrics which decline faster than they rise (and vice versa) by examining whether the metrics lie above the y=x line in Figure~\ref{fig:peaks_ascent_decay}.  For most metrics, the rate of rise and fall are roughly the same. However, minutes per game decay faster than they rise suggesting an athlete quickly becomes less durable as their career progresses, either playing less due to injury, capacity, or injury prevention protocols. Two-point shooting efficiency has the opposite asymmetry, improving rapidly early in a career and decaying more slowly post-peak. Ultimately, these multi-metric aging trends can help front offices and coaches understand and calibrate realistic performance expectations.

\begin{figure}[tbp]
    \centering
    \begin{subfigure}[b]{0.49\textwidth}
        \centering
        \includegraphics[width=\linewidth]{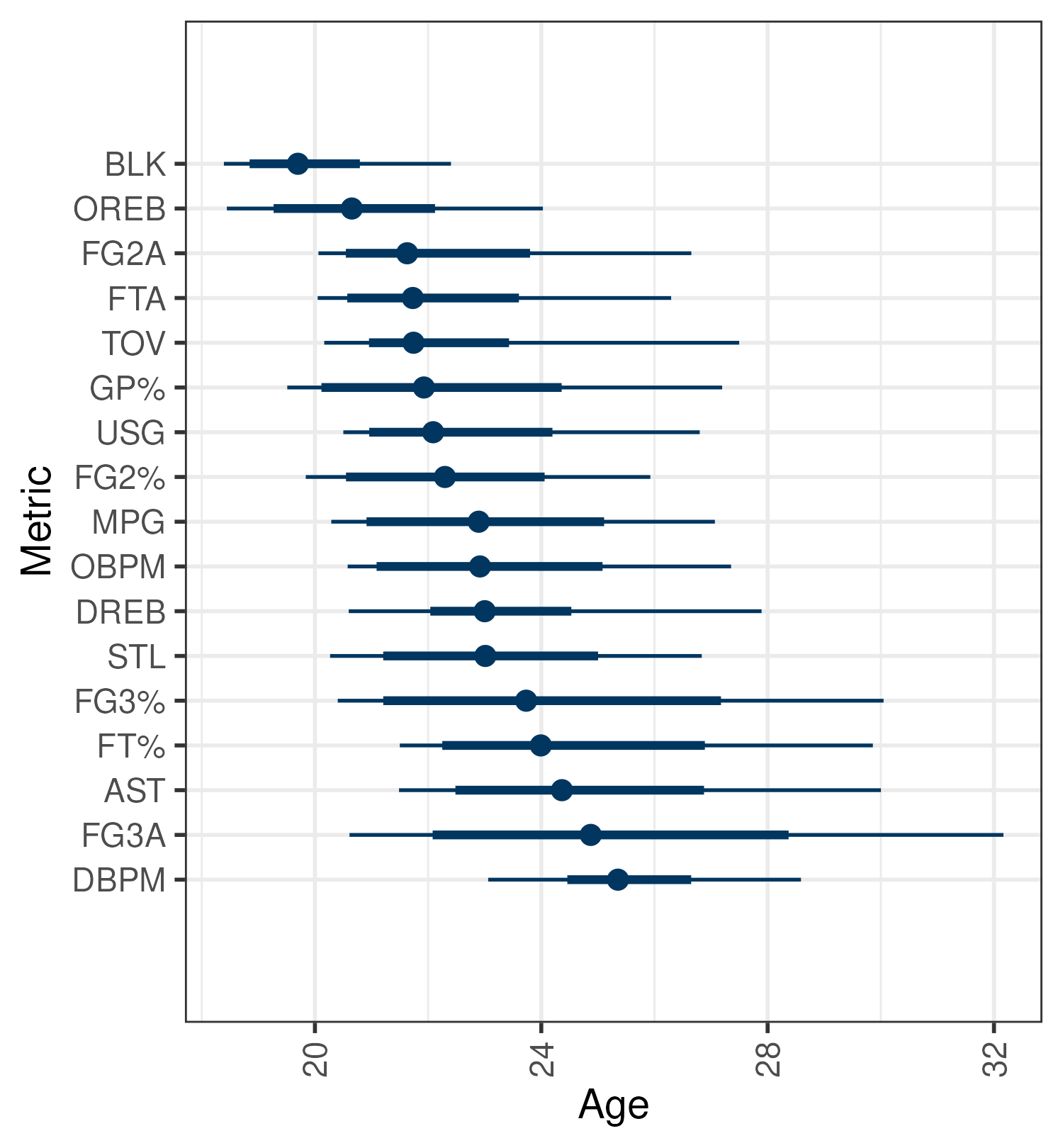}
        \caption{Peak age by metric}
        \label{fig:peaks}
    \end{subfigure}\hfill
    \begin{subfigure}[b]{0.49\textwidth}
        \centering
        \includegraphics[width=\linewidth]{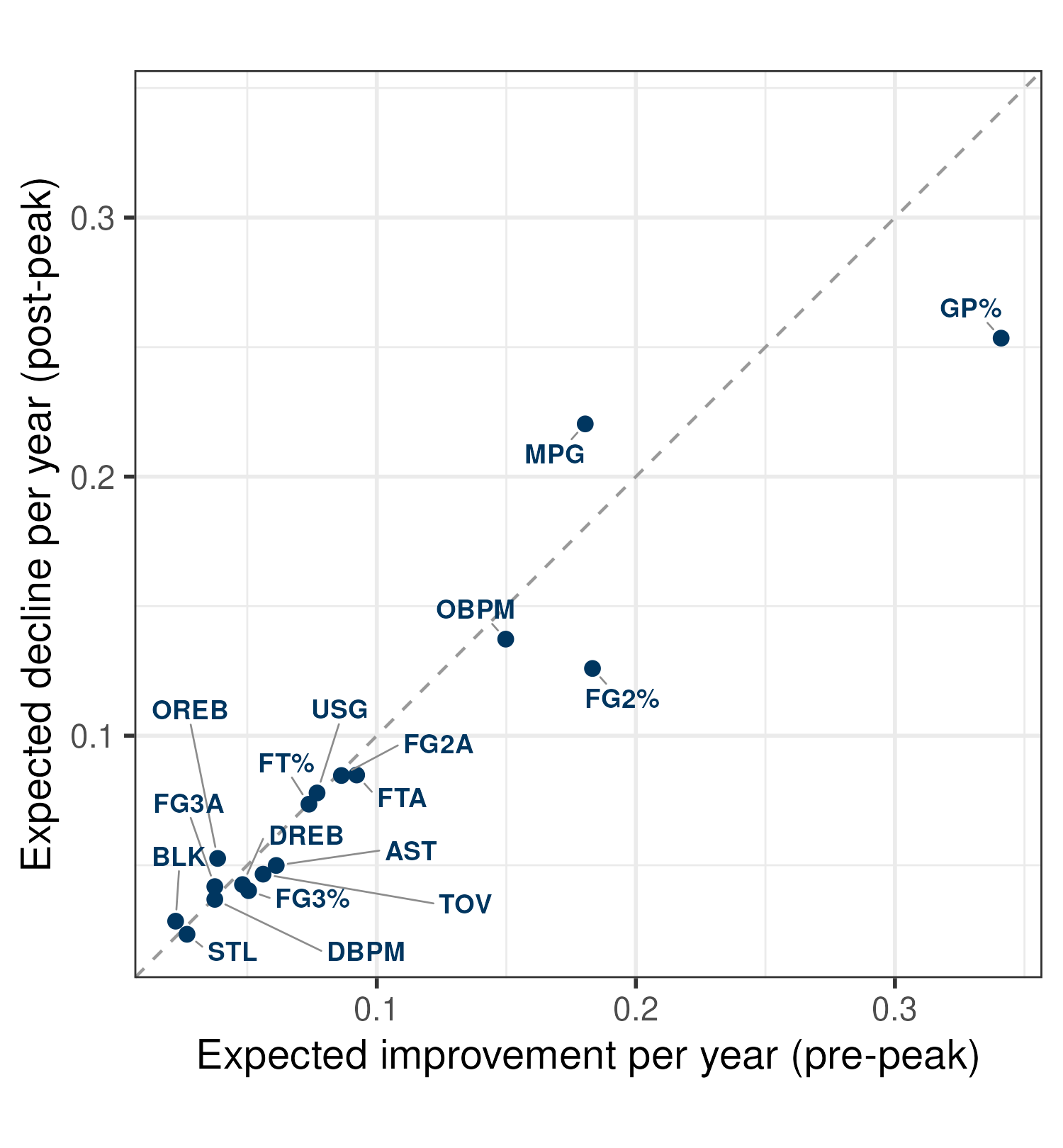}
        \caption{Expected pre-peak improvement vs.\ post-peak decline}
        \label{fig:peaks_ascent_decay}
    \end{subfigure}

    \caption{(a) Distribution of posterior means for the age of peak, separated by metric and ordered by population mean (dot), with 50\% (thick) and 95\% (thin) intervals. BLK and OREB peak earliest; the shooting metrics FG3\%, FT\%, and FG3A sit at the late end of the ordering. (b) Expected improvement per year over the four years before the peak against expected decline per year over the four years after, one point per metric. The horizontal coordinate is $\mathbb{E}_p[(f^{*} - f(a^{-}))/(t^{*} - a^{-})]/\mathrm{sd}_p(f^{*})$ and the vertical coordinate is $\mathbb{E}_p[(f^{*} - f(a^{+}))/(a^{+} - t^{*})]/\mathrm{sd}_p(f^{*})$, with $a^{-} = \max(t^{*}-4,\,18)$ and $a^{+} = \min(t^{*}+4,\,38)$. A metric on the dashed line rises and falls at the same normalised rate, one above it decays faster than it ascends. The largest asymmetry above the line belongs to MPG, which decays faster than it rose. OREB and BLK sit just above it. GP\% and FG2\% lie furthest below the line, rising faster than they decline.}
    \label{fig:peaks_combined}
\end{figure}

\section{Conclusion and Future Work}
\label{sec:conclusion}

We introduced a joint Bayesian model for multi-metric NBA career trajectories that combines two novel components 1) a Concave Process (CP) prior over production curves that encodes the biological constraint of concave aging, and 2) a latent variable model (GPLVM) that embeds all 17 metrics into a shared low-dimensional player representation.

Our results answer the three questions posed in Section~\ref{sec:contributions}.  First, the model renders peak age and peak value explicit posterior quantities, so a career is characterized not by a single summary but by when a player peaks, how high, and with what uncertainty. Players with limited data inherit credible trajectories from their nearest latent neighbors, allowing us to make multi-metric forecasts for rookies and young players.  Second, different aspects of performance are affected by aging in different way: rebounding and shot-blocking peak in a player's early twenties, shooting efficiency peaks in the mid-twenties, and three-point attempt volume continues rising into the early thirties.  Third, multi-task inference improves performance forecasts, emphasizing the value of multivariate predictive models in sports.

Several limitations should be noted.
First, the model is estimated on players who appeared in at least one NBA game, introducing survivor bias. Players who failed to reach or sustain NBA careers are unobserved, so the posteriors describe the distribution over NBA-caliber trajectories rather than the full talent distribution.
Second, the latent kernel is stationary, so similarity between players depends only on the distance between their embeddings and not on where in the latent space they sit.  A non-stationary kernel could allow for more flexible modeling.

There are also several interesting extensions and generalizations that may be considered.  
The most immediate is to apply the framework to other professional sports leagues where longitudinal performance data are available (e.g., soccer, hockey, baseball), testing whether CPLVM structure generalize beyond basketball.
In addition, our model could be used to estimate the impacts of injury by comparing counterfactual career trajectories to observed post-injury performance.  
Finally, integrating the trajectory model with contract-length and salary-cap data would allow front offices to assess not just athletic performance but economic value over the duration of a proposed contract.

\clearpage

\clearpage
{\sloppy
\bibliographystyle{plainnat}
\bibliography{references}
}

\clearpage
\section*{Supplementary Material}
\setcounter{figure}{0}\renewcommand{\thefigure}{S\arabic{figure}}\def\theHfigure{S\arabic{figure}}
\setcounter{table}{0}\renewcommand{\thetable}{S\arabic{table}}\def\theHtable{S\arabic{table}}
\setcounter{equation}{0}\renewcommand{\theequation}{S\arabic{equation}}\def\theHequation{S\arabic{equation}}
\setcounter{theorem}{0}\renewcommand{\thetheorem}{S\arabic{theorem}}\def\theHtheorem{S\arabic{theorem}}
\setcounter{proposition}{0}\renewcommand{\theproposition}{S\arabic{proposition}}\def\theHproposition{S\arabic{proposition}}
\setcounter{definition}{0}\renewcommand{\thedefinition}{S\arabic{definition}}\def\theHdefinition{S\arabic{definition}}
\input{appendix}

\end{document}

%% file: figs/model_plots/coverage/combined_all_summary.tex
\begin{table}[htbp]
\centering
\caption{Predictive performance across holdout schemes for the nested model chain, ordered least to most restricted. \textbf{Coverage}: holdout 95\% HDI coverage, bold closest to nominal. \textbf{ELPPD}: change in log predictive density per held-out observation from imposing that column's restriction on the model to its left; positive favours the restriction, and $\dagger$ marks $|\Delta| > 1.96\,\mathrm{se}_{\Delta}$. Performance metrics only.}
\label{tab:all_summary}
\resizebox{\ifdim\width>\linewidth\linewidth\else\width\fi}{!}{\begin{tabular}{lcccc}
\toprule
Holdout Scheme & GP & CP & Asym Quadratic & Quadratic \\
\midrule
\multicolumn{5}{l}{\textbf{95\% Coverage}} \\
Hold-out Last $k$ & 86.6\% & \textbf{90.7\%} & 90.4\% & 90.4\% \\
Hold-out First $k$ & 87.1\% & \textbf{91.8\%} & 91.2\% & 91.2\% \\
Random Interior & 84.5\% & \textbf{91.6\%} & 91.5\% & 91.5\% \\
Hold-out Peak & 87.6\% & \textbf{92.6\%} & 92.3\% & 92.3\% \\
Stratified Next $k$ & 88.0\% & \textbf{91.7\%} & 91.6\% & 91.4\% \\
\midrule
\multicolumn{5}{l}{\textbf{$\Delta$ELPPD}} \\
Hold-out Last $k$ & \emph{ref} & \textbf{+0.25$\dagger$} & -0.015$\dagger$ & -0.0038 \\
Hold-out First $k$ & \emph{ref} & \textbf{+0.31$\dagger$} & -0.012$\dagger$ & +0.0028 \\
Random Interior & \emph{ref} & \textbf{+0.46$\dagger$} & -0.044$\dagger$ & +0.018$\dagger$ \\
Hold-out Peak & \emph{ref} & \textbf{+0.32$\dagger$} & +0.0020 & +0.014$\dagger$ \\
Stratified Next $k$ & \emph{ref} & \textbf{+0.25$\dagger$} & -0.016$\dagger$ & -0.0087$\dagger$ \\
\bottomrule
\end{tabular}}
\end{table}

%% file: figs/nba_convex_max_tvrflvm_AR_posyear_long/stratified_next_k/mcmc/plots/latent_space/map/breakout_by_cohort.tex
\begin{table}[htbp]
  \centering\small
  \resizebox{\linewidth}{!}{\begin{tabular}{rlrlllrrr}
    \toprule
     & & & & \multicolumn{2}{c}{OBPM} & & & \\
    \cmidrule(lr){5-6}
    Rank & Player & Age & Pos. & Peak Value & Peak Age & P($>$2) & P($>$3) & P($>$4) \\
    \midrule
    \multicolumn{9}{l}{\textit{Entry 2025 (2 seasons)}} \\
    \midrule
    1 & Jared McCain & 21 & G & 1.8 {\scriptsize [-1.0, 5.2]} & 23.6 {\scriptsize [20.7, 27.0]} & 43\% & 21\% & 9\% \\
    2 & Zaccharie Risacher & 20 & F & 1.4 {\scriptsize [-1.6, 4.9]} & 23.4 {\scriptsize [20.7, 25.6]} & 34\% & 16\% & 6\% \\
    3 & Ron Holland & 20 & F & 0.8 {\scriptsize [-2.4, 4.3]} & 23.7 {\scriptsize [20.7, 26.3]} & 22\% & 9\% & 3\% \\
    4 & Stephon Castle & 21 & G & 0.6 {\scriptsize [-2.0, 3.7]} & 22.6 {\scriptsize [20.5, 25.0]} & 17\% & 6\% & 2\% \\
    5 & Dalton Knecht & 24 & F & 0.5 {\scriptsize [-1.7, 2.9]} & 24.1 {\scriptsize [21.4, 27.8]} & 9\% & 2\% & 0\% \\
    \midrule[0.4pt]
    \multicolumn{9}{l}{\textit{Entry 2024 (3 seasons)}} \\
    \midrule
    1 & Victor Wembanyama & 22 & C & 4.2 {\scriptsize [1.7, 7.3]} & 23.4 {\scriptsize [21.1, 25.9]} & 95\% & 80\% & 54\% \\
    2 & Anthony Black & 22 & G & 2.7 {\scriptsize [-1.5, 6.3]} & 26.3 {\scriptsize [22.5, 28.7]} & 65\% & 44\% & 24\% \\
    3 & Amen Thompson & 23 & F & 2.0 {\scriptsize [-0.2, 4.6]} & 23.9 {\scriptsize [21.0, 27.2]} & 46\% & 19\% & 6\% \\
    4 & Brandon Miller & 23 & F & 1.5 {\scriptsize [-0.7, 4.1]} & 23.7 {\scriptsize [21.3, 25.6]} & 34\% & 12\% & 3\% \\
    5 & Cam Whitmore & 21 & F & 1.4 {\scriptsize [-1.1, 4.2]} & 22.4 {\scriptsize [20.2, 25.3]} & 30\% & 11\% & 3\% \\
    \midrule[0.4pt]
    \multicolumn{9}{l}{\textit{Entry 2023 (4 seasons)}} \\
    \midrule
    1 & Jalen Williams & 24 & G & 5.3 {\scriptsize [2.7, 7.9]} & 26.5 {\scriptsize [24.6, 28.4]} & 99\% & 96\% & 84\% \\
    2 & Paolo Banchero & 23 & F & 4.1 {\scriptsize [2.0, 6.7]} & 23.5 {\scriptsize [21.9, 25.2]} & 97\% & 84\% & 53\% \\
    3 & Dyson Daniels & 22 & G & 2.7 {\scriptsize [-0.3, 5.8]} & 25.9 {\scriptsize [23.2, 28.3]} & 68\% & 42\% & 20\% \\
    4 & Christian Braun & 24 & G & 2.6 {\scriptsize [-0.5, 5.6]} & 27.5 {\scriptsize [23.1, 30.3]} & 65\% & 40\% & 18\% \\
    5 & Shaedon Sharpe & 22 & G & 2.4 {\scriptsize [-0.3, 5.2]} & 24.4 {\scriptsize [21.4, 26.5]} & 60\% & 32\% & 12\% \\
    \midrule[0.4pt]
    \multicolumn{9}{l}{\textit{Entry 2022 (5 seasons)}} \\
    \midrule
    1 & Cade Cunningham & 24 & G & 4.0 {\scriptsize [1.9, 6.5]} & 24.0 {\scriptsize [22.5, 26.1]} & 97\% & 84\% & 45\% \\
    2 & Trey Murphy III & 25 & F & 3.7 {\scriptsize [1.6, 6.3]} & 26.2 {\scriptsize [23.9, 28.9]} & 94\% & 71\% & 38\% \\
    3 & Franz Wagner & 24 & F & 3.6 {\scriptsize [1.5, 6.1]} & 24.7 {\scriptsize [23.2, 26.6]} & 92\% & 68\% & 34\% \\
    4 & Cam Thomas & 24 & G & 3.5 {\scriptsize [0.7, 6.6]} & 25.8 {\scriptsize [23.1, 28.5]} & 83\% & 60\% & 35\% \\
    5 & Evan Mobley & 24 & F & 3.3 {\scriptsize [1.1, 5.6]} & 23.9 {\scriptsize [21.8, 25.9]} & 88\% & 61\% & 23\% \\
    \midrule[0.4pt]
    \multicolumn{9}{l}{\textit{Entry 2021 (6 seasons)}} \\
    \midrule
    1 & Tyrese Haliburton & 25 & G & 6.4 {\scriptsize [4.3, 8.7]} & 25.7 {\scriptsize [23.7, 28.3]} & 100\% & 100\% & 99\% \\
    2 & Anthony Edwards & 24 & G & 4.5 {\scriptsize [2.6, 6.5]} & 24.0 {\scriptsize [22.7, 25.4]} & 100\% & 95\% & 72\% \\
    3 & Tyrese Maxey & 25 & G & 4.4 {\scriptsize [2.2, 6.7]} & 26.1 {\scriptsize [24.2, 28.1]} & 98\% & 89\% & 62\% \\
    4 & LaMelo Ball & 24 & G & 4.2 {\scriptsize [2.4, 6.4]} & 23.0 {\scriptsize [20.8, 25.2]} & 99\% & 90\% & 55\% \\
    5 & Desmond Bane & 27 & G & 3.8 {\scriptsize [2.0, 6.1]} & 27.2 {\scriptsize [25.1, 29.5]} & 97\% & 76\% & 37\% \\
    \midrule[0.4pt]
    \multicolumn{9}{l}{\textit{Entry 2020 (7 seasons)}} \\
    \midrule
    1 & Zion Williamson & 25 & F & 5.2 {\scriptsize [3.6, 7.2]} & 22.8 {\scriptsize [20.9, 25.4]} & 100\% & 100\% & 93\% \\
    2 & Ja Morant & 26 & G & 4.6 {\scriptsize [2.7, 6.1]} & 23.3 {\scriptsize [21.9, 25.1]} & 100\% & 96\% & 77\% \\
    3 & Tyler Herro & 26 & G & 3.4 {\scriptsize [1.3, 5.2]} & 25.4 {\scriptsize [23.8, 27.6]} & 89\% & 69\% & 24\% \\
    4 & Darius Garland & 26 & G & 3.3 {\scriptsize [0.6, 5.7]} & 25.5 {\scriptsize [23.2, 27.6]} & 83\% & 67\% & 30\% \\
    5 & Cameron Johnson & 29 & F & 2.6 {\scriptsize [1.1, 4.3]} & 27.1 {\scriptsize [25.2, 29.1]} & 73\% & 29\% & 5\% \\
    \bottomrule
  \end{tabular}}
  \caption{Top 5 players by posterior mean peak OBPM per entry cohort. Peak Value and Peak Age are posterior means with 95\% credible intervals in brackets. P($>$2), P($>$3), P($>$4) are posterior probabilities that peak OBPM exceeds the threshold; empirically, 19\% of players peaked above OBPM\,$=$\,2, 12\% above OBPM\,$=$\,3, and 7\% above OBPM\,$=$\,4 (career peak per player).}
  \label{tab:breakout_cohort}
\end{table}

%% file: appendix.tex
\appendix

\section{Enforcing Concavity}

\label{appendix:convexity}
\subsection{Validity of Prior}
\label{appendix:convexity_proof}
\begin{proof}
The map $\Phi(a,b,g)(t) = a + b(t-t_0) - \int_{t_0}^t\int_{t_0}^s g(z)^2\,dz\,ds$ is continuous since it composes the bounded, continuous operations (addition, squaring, integrating).  A continuous
transformation of a distribution is a valid distribution over the image of the transformation.  That
image consists of concave functions, since $\Phi(a,b,g)^{''} = -g^2 \le 0$.
\end{proof}

\subsection{Univariate Case}
\label{appendix:univariate_concavity}
Considering we only wish our function $f(t)$ to be concave in the $t$ argument, we wish to find all set of $f(t)$ such that $\frac{\partial ^2 f(t)}{\partial t^2} \geq 0$. As such, our desired function has the form
\begin{align}
\label{eq:concave_uni_rep}
    f(t) = a + b(t-t_0) - \int_{t_0}^t \int_{t_0}^s \frac{\partial ^2 f(z)}{\partial z^2}dzds 
\end{align}
where $a, b$ are some constants of integration with respect to $t$. We show in the subsequent section how these constants of integration can be defined to modulate the shape of $f(t)$.

Following the line of reasoning from \cite{convexityconstraints}, we assume that $\frac{\partial ^2 f(t)}{\partial t^2}$ can be modeled by $[g(t)]^2$, where $g(t) \sim \mathcal{GP}(0, K_t)$. 
In order to enforce concavity across the $t$ dimension, we use the Hilbert Space Gaussian Process (HSGP) approximation mentioned in \cite{convexityconstraints} to approximate $g(t) \approx \sum_{l=1}^K S(\sqrt{\lambda_l})^{1/2} \alpha_l \psi_l(t)$.

Consequently,
\begin{align}
\label{eq:deriv_approx_uni_gp}
g(t)^2 \approx \sum_{l=1}^K \sum_{z=1}^K S(\sqrt{\lambda_z})^{1/2} S(\sqrt{\lambda_l})^{1/2} \alpha_l \alpha_z \psi_l(t) \psi_z(t)\nonumber \\
 \approx  \sum_{l=1}^K \sum_{z=1}^K\tilde{\alpha_l} \tilde{\alpha_z} \psi_l(t)\psi_z(t) 
\end{align}

after letting $\tilde{\alpha}_l = S(\sqrt{\lambda_l})^{1/2} \alpha_l $.
As such, substituting \eqref{eq:deriv_approx_uni_gp} into \eqref{eq:concave_uni_rep} gives
\begin{align}
    & f(t) \approx a + b(t-t_0) - \sum_{l=1}^K \sum_{z=1}^K\tilde{\alpha_l} \tilde{\alpha_z}  \int_{t_0}^t \int_{t_0}^s \psi_l(w)\psi_z(w)  dwds
\end{align}

Rewriting the above in terms of matrix products, and setting $t_0 = -L$,  yields
\begin{align}
\label{eq:final_gp_app_uni}
    f(t) \approx a + b(t+L) -  \mathbf{\tilde{\alpha}}^T \Psi(t) \mathbf{\tilde{\alpha}} 
\end{align}

Here, $\mathbf{\tilde{\alpha}}$ is a $K$ vector with entries given by $\tilde{\alpha_l}$. 
The entries of the $K \times K$ matrix $\Psi(t)$ are given by
\begin{align}
\label{eq:phi_app}
\Psi_{lz}(t) =
\begin{cases} \frac{cos(\epsilon_{lz}^{+}(t+L)) - 1}{2L(\epsilon_{lz}^{+})^2} +
\frac{(t+L)^2}{4L}  & \text{if } l = z \\
\frac{cos(\epsilon_{lz}^{+}(t+L)) - 1}{2L(\epsilon_{lz}^{+})^2} - \frac{cos(\epsilon_{lz}^{-}(t+L)) - 1}{2L(\epsilon_{lz}^{-})^2} & \text{if } l \neq z
\end{cases}
\end{align}

Here, $\epsilon_{lz}^{+} = \sqrt{\lambda_l} + \sqrt{\lambda_z}$, and $\epsilon_{lz}^{-} = \sqrt{\lambda_l} - \sqrt{\lambda_z}$. Note that since the eigen-functions $\psi_l(t)$ are sinusoidal, the integral in \eqref{eq:concave_uni_rep} can be computed in closed form.

The concavity prior is placed on the latent mean $\mu_{ptk}$ in the linear predictor space; it does not, in general, guarantee concavity of the observable response after applying the inverse link.  However, the scientifically relevant properties are fully preserved.  All link functions used in the model are strictly monotone increasing.  For any strictly monotone $g^{-1}$,
\[
  \arg\max_t \; g^{-1}(\mu_{ptk}(t)) \;=\; \arg\max_t \; \mu_{ptk}(t),
\]
so the peak age $t^*_k$ is identical in the latent and observable spaces.  The peak value in the observable space, $g^{-1}(\mu_{ptk}(t^*_k))$, is a monotone function of the latent peak, so the rank ordering of peak values across players is also preserved.

\subsection{Multivariate Case}
\label{appendix:multivariate_concavity}
Considering we only wish our function $f(x,t)$ to be concave in the $t$ argument, but still incorporate information from another set of dimensions encoded by $x$, we wish to find all set of $f(x,t)$ such that $\frac{\partial ^2 f(x,t)}{\partial t^2} \geq 0$. As such, our desired function has the form
\begin{align}
\label{eq:concave_multi_rep_app}
    f(x,t) = \beta_0(x) + \beta_1(x)(t-t_0) - \int_{t_0}^t \int_{t_0}^s \frac{\partial ^2 f(x,z)}{\partial z^2}dzds
\end{align}
where $\beta_0(x), \beta_1(x)$ are constants of integration with respect to $t$. We show in the subsequent section how these constants of integration can be defined to modulate the shape of $f(x,t)$.

Following the line of reasoning from \cite{convexityconstraints}, we assume that $\frac{\partial ^2 f(x,t)}{\partial t^2}$ can be modeled by $[g(x,t)]^2$, where $g(x,t) \sim \mathcal{GP}(0, K_x \otimes K_t)$. Note that this assumes that the time variable $t$ is separate from our $x$ variable . 

Using a Random Fourier Feature (Appendix~\ref{sec:rff}) projection $\phi_i(x)$ for the $x$ dimension, we can write $g(x,t) \approx \sum_{i=1}^M \beta_i(t) \phi_i(x)$. Here, $\beta_i(t)$ are $t$-dependent coefficients of the basis expansion. In order to enforce concavity across the $t$ dimension, we use the Hilbert Space Gaussian Process (HSGP) approximation mentioned in \cite{convexityconstraints} to approximate $\beta_i(t) \approx \sum_{l=1}^K S(\sqrt{\lambda_l})^{1/2} \alpha_l^i \psi_l(t)$. Note that this approximate representation preserves the separability of the covariance structure of our desired $g(x,t)$
(Appendix~\ref{appendix:separability_properties}).

If we let $\tilde{\alpha_l^i} = \alpha_l^i S(\sqrt{\lambda_l})^{1/2}$  , then $g(x,t) \approx \sum_{i=1}^M \phi_i(x) ( \sum_{l=1}^K \tilde{\alpha_l^i} \psi_l(t))$. 

Consequently,
\begin{align}
\label{eq:deriv_approx_gp}
g(x,t)^2 \approx \sum_{i=1}^M \sum_{j=1}^M \phi_j(x)\phi_i(x) \beta_j(t) \beta_i(t) \nonumber \\
\approx \sum_{i=1}^M \sum_{j=1}^M \phi_j(x)\phi_i(x) (\sum_{l=1}^K \tilde{\alpha_l^i} \psi_l(t)) (\sum_{l=1}^K \tilde{\alpha_l^j} \psi_l(t)) \nonumber \\  
 \approx \sum_{i=1}^M \sum_{j=1}^M \phi_j(x)\phi_i(x) (\sum_{l=1}^K \sum_{z=1}^K\tilde{\alpha_l^j} \tilde{\alpha_z^i} \psi_l(t)\psi_z(t) ) 
\end{align}

As such, substituting \eqref{eq:deriv_approx_gp} into Equation~(17) of the main text gives
\begin{align}
    & f(x,t) \approx \beta_0(x) + \beta_1(x)(t-t_0) -  \nonumber \\
    & \sum_{i=1}^M \sum_{j=1}^M \phi_j(x)\phi_i(x) (\sum_{l=1}^K \sum_{z=1}^K\tilde{\alpha_l^j} \tilde{\alpha_z^i}  \int_{t_0}^t \int_{t_0}^s \psi_l(w)\psi_z(w)  dwds)
\end{align}

Rewriting the above in terms of matrix products, and setting $t_0 = -L$,  yields
\begin{align}
\label{eq:final_gp_app}
    f(x,t) \approx \beta_0(x) + \beta_1(x)(t+L) -  \mathbf{\phi}(x)^T\mathbf{\tilde{A}}^T \Psi(t) \mathbf{\tilde{A}} \mathbf{\phi}(x)
\end{align}

Here, $\mathbf{\tilde{A}}$ is a $M \times K$ matrix with entries given by $\tilde{\alpha_l^j}$. 
The entries of the $K \times K$ matrix $\Psi(t)$ are given by Equation~(14) of the main text, from the univariate case.

\subsection{Controlling the Peak Behavior}
\label{appendix:peak_value}
We may wish to place a prior on the location of the peak of our curve, as well as the value that this peak takes. Thus, our formulation begins with the following constraints:
\begin{enumerate}
    \item $f(x,t_{max}) = f_{max}$
    \item $\frac{d}{dt} f(x,t) |_{t = t_{max}} = 0$
\end{enumerate}

The first expectation illustrates that at our peak $t_{max}$, our function takes on the value $c$. The second expectation requires that our first derivative with respect to $t$ is 0 at $t_{max}$. Since our function is concave by construction, we do not need to constrain the behavior of the second derivative. Note here that $t_{max}$ and $f_{max}$ are specified apriori.

Taking the partial derivative of the main text's Equation~(18) with respect to $t$, and setting it equal to 0 yields the following requirement
\begin{align}
\label{eq:first_deriv_constraint}
    \beta_1(x) = \mathbf{\phi}(x)^T\mathbf{\tilde{A}}^T \Psi^{'}(t_{max}) \mathbf{\tilde{A}} \mathbf{\phi}(x)
\end{align}
Here, $\Psi^{'}$ is the first derivative of the entries of the main text's Equation~(14) with respect to $t$. 
Now, it remains to satisfy the peak value constraint. Formally,

\begin{align}
\label{eq:peak_value_constraint}
    \beta_0(x) = f_{max} - \beta_1(x)(t_{max} - t_0) + \mathbf{\phi}(x)^T \mathbf{\tilde{A}}^T \Psi(t_{max}) \mathbf{\tilde{A}} \mathbf{\phi}(x)
\end{align}
If we combine \eqref{eq:first_deriv_constraint} and \eqref{eq:peak_value_constraint}, and plug the resulting values for $\beta_0(x)$ and $\beta_1(x)$ into Equation~(18) of the main text, we are met with the following theorem.

\begin{align}
\label{eq:max_constraint_gp_app}
    f(x,t) \approx f_{max}  +  \mathbf{\phi}(x)^T\mathbf{\tilde{A}}^T [ \Psi(t_{max}) - \Psi(t) + \Psi^{'}(t_{max})(t - t_{max}) ] \mathbf{\tilde{A}} \mathbf{\phi}(x)
\end{align}
Note that when $t = t_{max}$, we recover $f_{max}$ as desired.

Subsequently, letting $f_{max}(x) , t_{max}(x) \sim GP(0, K_x)$ allows us to pool information across the $x$ dimension, allowing similar players to have similar peaks / peak-values.

\subsection{Convergence of the Push-Forward HSGP Approximation}

\begin{proposition}[Convergence of the HSGP approximation]
\label{prop:hsgp_convergence_proof}
Let $K$ be a stationary covariance function with spectral density $S(\cdot)$, let $g \sim GP(0,K)$, and fix
an interval of interest $[0,T]$.  For a domain $[-L,L] \supset [0,T]$ let
$\{\psi_i,\lambda_i\}_{i=1}^\infty$ be the Dirichlet Laplacian eigenpairs on $[-L,L]$, write
$\omega_i = \sqrt{\lambda_i}$, and let
\[
g_{M,L}(t)
=
\sum_{i=1}^M
\sqrt{S(\omega_i)}\,\alpha_i \psi_i(t),
\qquad
\alpha_i \stackrel{iid}{\sim} \mathcal{N}(0,1),
\]
be the HSGP approximation of \cite{Solin_2019}.
Define
\[
f(t)
=
\int_0^t \int_0^s g(z)^2\,dz\,ds,
\qquad
f_{M,L}(t)
=
\int_0^t \int_0^s g_{M,L}(z)^2\,dz\,ds.
\]
Then, as $L \to \infty$ with $M/L \to \infty$,
\[
f_{M,L} \to f
\quad \text{in distribution on } C([0,T]).
\]
\end{proposition}

\begin{proof}
By Theorem 1 of \cite{Solin_2019}, $K_{M,L} \to K$ uniformly.  Since both processes are centred
Gaussian, this gives $g_{M,L} \to g$ in distribution on $L^2([0,T])$.  Squaring is continuous from
$L^2$ to $L^1$ and double integration is bounded linear from $L^1$ into $C([0,T])$, so the
continuous mapping theorem gives $f_{M,L} \to f$ in distribution.
\end{proof}

\clearpage

\clearpage
\section{Model Specifications}

\subsection{Accounting for Seasonal Trends Across Metrics}
\label{sec:detrending}
Three-point volume, pace, and positional usage have all changed substantially over the seasons spanned by the data. Left unmodeled, this era variation is absorbed into the aging function, and $f_k$ conflates how a player changes with how the league changed.

Let $c_p \in \{\text{G}, \text{F}, \text{C}\}$ denote the position group of player $p$ --- guard, forward, or center --- and $s_{pt}$ the calendar season in which player $p$ is aged $t$.  Writing $\bar{y}_k(c,s)$ for the observed average of metric $k$ over all players of position group $c$ in season $s$, the offset $o_{ptk} = g_k(\bar{y}_k(c,s))$ enters each metric's linear predictor through that metric's own link. For ages projected beyond the observed seasons, the offset is set to the nearest available season.

The grouping is by position and season rather than season alone because league trends such as three point shooting did not effect all positions equally. For example, big men saw a larger proportionate increase in three point attempt rate than guards. The curve
$f_k$ then represents production relative to a player's positional peers in the same season.

\subsection{Latent Embedding Prior}
\label{sec:embedding_prior}

Let
\[
  Z_p = \bigl[-\log d_p,\; h_p,\;\mathbf{1}[\mathrm{pos}_p = F],\;\mathbf{1}[\mathrm{pos}_p = G]\bigr]^\top \in \mathbb{R}^4
\]
denote four player-level age agnostic covariates: negative log draft position $-\log d_p$ (a proxy for talent), height $h_p$, and two indicator variables for positional group (Forward and Guard, with Center as the baseline category).  Players without a recorded draft position are assigned a rank of $P{+}1$, where $P$ is the total number of players in the cohort, so that $-\log d_p$ defaults to the minimum value after standardization.  We introduce a learned projection $W \in \mathbb{R}^{4 \times q}$ with scale $\sigma_W$ and a residual scale $\sigma_X$, and place the prior
\begin{align}
\label{eq:structured_prior_full}
  X_p &= Z_p^\top W + \sigma_X\, \varepsilon_p, \quad
  \varepsilon_p \sim t_\nu\!\bigl(0,\, \sqrt{(\nu-2)/\nu}\; I_q\bigr), \quad \nu = 4, \nonumber\\
  W &\sim \mathcal{N}(0,\, \sigma_W^2\, I_{4\times q}).
\end{align}
The residual is Student-$t$ rather than Gaussian to account for the probability of outlier talent. Scales are learned at the MAP and held fixed during
posterior sampling ($\sigma_X = 1.90$, $\sigma_W = 0.16$). 
The prior mean $Z_p^\top W$ pulls a player's latent position toward the subspace predicted by draft standing, physical size, and positional role; the residual scale $\sigma_X$  controls how tightly the data must pull $X_p$ away from that prior mean.  When $\sigma_W \to 0$ the prior reduces to the isotropic case $\sigma_X \varepsilon_p$.  Finally, $W$ is shared across all players and learned from the data.

\subsection{Procrustes Alignment of Posterior Draws}
\label{sec:procrustes}

Because the latent embedding $X_p \in \mathbb{R}^q$ is only identified up to orthogonal rotation of its learned coordinates, raw MCMC draws from different chains occupy different rotation frames and cannot be directly compared or averaged. There are different approaches for dealing with nonidentifiability of factors across multiple MCMC samples \citep{procrustes, poworoznek2024efficientlyresolvingrotationalambiguity}. To enable coherent posterior summaries, we apply Procrustes alignment to each (chain, draw) pair: given the MAP estimate $\hat{X}$ as a fixed reference, we find the orthogonal matrix $R^{(s)}$ minimising $\|\hat{X} - X^{(s)} R^{(s)}\|_F$ via singular value decomposition, and replace each sampled $X^{(s)}$ with the aligned version $X^{(s)} R^{(s)}$ before any downstream computation.  

\subsection{Career Duration Formulation}
\label{sec:survival_linear_map}
We model career exit with a player-specific Gompertz proportional-hazards model. We select the Gompertz model since its hazard rises monotonically with age, and has historically been used to model human aging.  Because player $p$ is only observed from entrance age $\alpha_p$ onward, the exit-time likelihood is left-truncated at $\alpha_p$ (and right-censored for players still active at the end of the observation window).  The Gompertz hazard and cumulative hazard are
\begin{equation}
    h_p(a) = \eta_p\,e^{\gamma_p a}, \qquad
    H_p(a) = \frac{\eta_p}{\gamma_p}\bigl(e^{\gamma_p a} - 1\bigr),
\end{equation}
and the left-truncated log-likelihood contributions for exit age $T_p$ are
\begin{equation}
    \ell_p =
    \begin{cases}
        \log h_p(T_p) - \bigl[H_p(T_p) - H_p(\alpha_p)\bigr] & \text{(observed exit)}, \\
        -\bigl[H_p(T_p) - H_p(\alpha_p)\bigr] & \text{(right-censored)},
    \end{cases}
    \label{eq:surv_lik}
\end{equation}
where $\log h_p(a) = \log\eta_p + \gamma_p a$ and the truncated increment is $H_p(T_p) - H_p(\alpha_p) = (\eta_p/\gamma_p)\bigl(e^{\gamma_p T_p} - e^{\gamma_p \alpha_p}\bigr)$.  Both parameters are linked to the latent embedding on the log scale, through the same feature-map projection $\phi(X_p)$ used by the performance likelihood. The latent embedding thereby ties career longevity to the same low-dimensional representation that drives on-court performance:
\begin{align}
    \log \eta_p   &= \eta_{\text{glob}} + \sigma_{\text{exit}}\,\phi(X_p)^\top \psi_{\text{exit}}, \label{eq:eta} \\
    \log \gamma_p &= \gamma_{\text{glob}} + \phi(X_p)^\top \psi_{\text{rate}}, \label{eq:gamma}
\end{align} Here $\eta_{\text{glob}} \sim \mathcal{N}(\log 0.04,\, 0.5)$ and $\gamma_{\text{glob}} \sim \mathcal{N}(\log 0.15,\, 0.3)$ are global log baseline-hazard and log aging-rate offsets, $\sigma_{\text{exit}} \sim \mathrm{HalfNormal}(0.5)$ controls the player-specific spread in baseline hazard, and $\psi_{\text{exit}} \sim \mathcal{N}(0, I)$, $\psi_{\text{rate}} \sim \mathcal{N}(0,\, 0.1^2 I)$, with dimension matching $\phi(X_p)$.  The exponential links keep $\eta_p, \gamma_p > 0$, so every player's exit hazard is increasing in age. 

\subsection{Random Fourier Feature Latent Projection}
\label{sec:rff}
 We model player-player correlation with a \emph{random Fourier feature} (RFF) expansion \citep{rffs}, which induces a stationary, automatic-relevance-determination (ARD) squared-exponential kernel over the latent space.  This is the random Fourier feature latent variable model (RFLVM) of \cite{gundersen2021latent}, together with its fully Bayesian extension \cite{zhang2023bayesiannonlinearlatentvariable}.

Draw $m$ spectral frequencies collected into $\Omega \in \mathbb{R}^{m \times q}$ with $\Omega_{ij} \sim \mathcal{N}(0,1)$, together with a per-dimension lengthscale $\ell = (\ell_1, \dots, \ell_q)$, $\ell_j \sim \mathrm{HalfNormal}(1)$.  Writing $\tilde\Omega = \Omega\,\mathrm{diag}\!\bigl(\sqrt{\ell_1}, \dots, \sqrt{\ell_q}\bigr)$, define the unit-norm feature map
\begin{equation}
  \phi(X_p) = \frac{1}{\sqrt{m}}\Bigl[\cos(\tilde\Omega X_p)^\top,\; \sin(\tilde\Omega X_p)^\top\Bigr]^\top \in \mathbb{R}^{2m},
  \qquad \lVert \phi(X_p) \rVert_2 = 1.
  \label{eq:rff_features}
\end{equation}
By Bochner's theorem the feature inner product is an unbiased Monte-Carlo estimator of the ARD squared-exponential kernel \citep{rffs},
\begin{equation}
  \phi(X_p)^\top \phi(X_{p'}) \;\xrightarrow[\;m \to \infty\;]{}\;
  k(X_p, X_{p'}) = \exp\!\Bigl(-\tfrac{1}{2}\textstyle\sum_{j=1}^{q} \ell_j\,(X_{pj} - X_{p'j})^2\Bigr),
  \label{eq:rff_kernel}
\end{equation}
so $\ell_j$ is the relevance (inverse squared bandwidth) of latent dimension $j$. The peak-age and peak-value weights become $\boldsymbol\gamma^k_{\text{age}}, \boldsymbol\gamma^k_{\text{val}} \in \mathbb{R}^{2m}$, the concave HSGP basis weights become $\mathbf{A}_k \in \mathbb{R}^{2m \times M}$, and the Gompertz links of Equations~\ref{eq:eta}--\ref{eq:gamma} specialise to
\begin{equation}
  \log \eta_p = \eta_{\text{glob}} + \sigma_{\text{exit}}\,\phi(X_p)^\top \boldsymbol\psi_{\text{exit}},
  \qquad
  \log \gamma_p = \gamma_{\text{glob}} + \phi(X_p)^\top \boldsymbol\psi_{\text{rate}},
  \label{eq:rff_surv}
\end{equation}
with $\boldsymbol\psi_{\text{exit}}, \boldsymbol\psi_{\text{rate}} \in \mathbb{R}^{2m}$. 

We keep the structured prior on $X_p$ (Equation~\ref{eq:structured_prior_full}) unchanged. The projection remains shared across all metrics and the survival model.  The frequencies $\Omega$ are drawn once and are learned rather than held at their random initialisation. We initialise $\Omega$ at its MAP estimate and continue to sample it during MCMC, while the lengthscale $\ell$ is learned at the MAP estimate and plugged in as a fixed value during posterior sampling.  We use $m = 50$ random features (hence $2m = 100$ feature dimensions) over the $q = 10$ latent dimensions. We summarise the RFF embedding through the ARD-rescaled coordinates $X_p \odot \sqrt{\ell}$, which place inter-player distances on the scale of the induced kernel~\eqref{eq:rff_kernel}.

\subsection{Model Priors}
\label{sec:model_priors}

Table~\ref{tab:prior_spec} lists the prior on each parameter, its role in the model, and whether it is sampled during MCMC or held at its MAP estimate.

\input{figs/model_plots/coverage/prior_spec.tex}

\clearpage
\section{Hilbert Space Gaussian Processes}
\label{appendix:HSGPs}
Traditionally, sampling from a full Gaussian Process is a computationally expensive procedure, involving inverting a typically large $N \times N$ covariance matrix $K$. \cite{Solin_2019} showed, however, that a reduced rank approximation of traditional stationary covariance matrices can be utilized instead, reducing computational burden. The resulting approximation represents the function $f(x)\approx \sum_{l=1}^K S(\sqrt{\lambda_l})^{1/2} \alpha_l \psi_l(x)$ as a finite, weighted, linear combination of eigenvectors, $\psi_l(x)$, of the Laplacian on the domain of $x$, with the weights $\alpha_l \sim N(0,S(\sqrt{\lambda_l}))$, arising as a function of the spectral density $S(\sqrt{\lambda_l})$ evaluated at the associated eigenvalue $\sqrt{\lambda_l}$ of the stationary covariance matrix $K$. If $\mathcal{K}$ is the standard squared exponential kernel, $S(\sqrt{\lambda_l}) = s^2\sqrt{2\pi}\omega \text{exp}(\frac{-\omega^2\lambda_l}{2})$ is the spectral density of the squared exponential with length-scale parameter $\omega$ and variance parameter $s$. As we increase the number of dimensions, this approximation does not scale well, as eigenvalues for each d-tuple are needed. As such, this approach is useful for one-dimensional cases, where the domain $\Omega$ is bounded on  $[-L, L]$. 
In the one dimensional Dirichlet boundary condition case, the eigenvalues and eigenvectors are $\lambda_j = (\frac{j \pi}{2L})^2$ and $\psi_l(x) = \sqrt{\frac{1}{L}}sin(\sqrt{\lambda_l}(L + x))$ respectively. For a full treatment of this approach, please refer to \cite{Solin_2019}.

\clearpage
\section{Approximating a Separable Kernel}
\label{appendix:separability_properties}

We make the claim that
$g(x,t) \approx \sum_{i=1}^M \phi_i(x) ( \sum_{l=1}^K \alpha_l^i \psi_l(t))$ has a separable covariance structure in $x$ and $t$. That is, $Cov( g(x,t) , g(x^{'},t^{'} ))$ can be factored into a product of kernels $K(t,t^{'})K(x,x^{'})$. 

First, we assume that $\alpha_l^i \sim N(0, S(\sqrt{\lambda_l}))$. Furthermore,
\begin{align}
\label{eq:cov_separable}
Cov(\alpha_l^i, \alpha_z^k) = 
\begin{cases}
S(\sqrt{\lambda_l})& \text{if} \quad l = z, k = i \\
0 & \text{o.w}
\end{cases}
\end{align}
These assumptions align with the formulation outlined in \cite{Solin_2019}.

First, we will evaluate $Cov(\beta_i(t), \beta_i(t^{'}))$,
where $\beta_i(t) = \sum_{l=1}^K \alpha_l^i \psi_l(t)$.

\begin{align}
    Cov(\beta_i(t), \beta_i(t^{'})) = Cov(\sum_{l=1}^K \alpha_l^i \psi_l(t), \sum_{l=1}^K \alpha_l^i \psi_l(t^{'})) 
\end{align}
Using \eqref{eq:cov_separable}, we reach that 
\begin{align}
    Cov(\beta_i(t), \beta_i(t^{'})) = \sum_{l=1}^K  \psi_l(t^{'}) \psi_l(t) S(\sqrt{\lambda_l})
\end{align}
From \cite{Solin_2019}, $\sum_{l=1}^K  \psi_l(t^{'}) \psi_l(t) S(\sqrt{\lambda_l}) = \tilde{K}(t,t^{'})$ converges to $ K(t, t^{'})$ uniformly in the univariate case. Since we only consider the time dimension, this suffices for our purposes. 

Now, we can evaluate $Cov( g(x,t) , g(x^{'},t^{'} ))$.
\begin{align}
    Cov( g(x,t) , g(x^{'},t^{'} )) = Cov(\sum_{i=1}^M \phi_i(x)  \beta_i(t) , \sum_{i=1}^M \phi_i(x^{'})  \beta_i(t^{'}) )
\end{align}
Note that from \eqref{eq:cov_separable},
\begin{align}
\label{eq:cov_separable_2}
Cov(\beta_i(t), \beta_j(t)) = 
\begin{cases}
\sum_{l=1}^K S(\sqrt{\lambda_l}) \psi_l(t)^2& \text{if} \quad j = i \\
0 & \text{o.w}
\end{cases}
\end{align}
Thus,
\begin{align}
    Cov\!\left( g(x,t),\, g(x^{'},t^{'}) \right)
    &= \sum_{i=1}^M \phi(x)\phi(x^{'}) Cov\!\left(\beta_i(t), \beta_i(t^{'})\right) \nonumber\\
    &= \tilde{K}(t,t') \sum_{i=1}^M \phi(x)\phi(x^{'})
\end{align}
{\sloppy In our implementation, $\phi_i(x)$ is a linear (identity) projection, so $\sum_{i=1}^M \phi(x)\phi(x^{'}) = K_{\mathrm{lin}}(x,x^{'})$ exactly. The product $\tilde{K}(t,t^{'}) K_{\mathrm{lin}}(x,x^{'})$ therefore gives an exact separable covariance structure in the $x$ dimension, while the HSGP approximation $\tilde{K}(t,t^{'})$ converges uniformly to $K(t,t^{'})$ from \cite{Solin_2019}.\par}

{\sloppy For the random Fourier feature projection (Appendix~\ref{sec:rff}), the $x$-factor $\sum_{i=1}^{2m}\phi_i(x)\phi_i(x^{'}) = \phi(x)^\top\phi(x^{'})$ is instead a random-feature approximation of the ARD kernel that converges as the number of features grows \citep{rffs}.\par}

\clearpage

\section{MCMC Convergence Diagnostics}
\label{sec:convergence}

Table~\ref{tab:convergence} summarizes split-$\hat{R}$ and tail effective sample size for the peak ages, peak values, concave curves and latent embeddings.

\input{figs/nba_convex_max_tvrflvm_AR_posyear_long/stratified_next_k/mcmc/plots/mcmc/convergence_table_paper.tex}

\clearpage
\section{Predictive Model Comparison}
\label{sec:elppd_methodology}

\subsection{The value of joint modelling}
Restricting the model to only use OBPM for predictive performance worsens the model's performance when compared to OBPM predicted from using all metrics (Table~\ref{tab:obpm_ablation}).
\input{figs/model_plots/coverage/obpm_ablation.tex}

\clearpage
\section{Model Validation}
\label{sec:model_validation}
\label{sec:appendix_validation}

The five holdout schemes are implemented by selecting a 20\% subset of players and masking their observations as follows. 

\begin{description}
  \item[Hold-out Last $k$ ($k{=}2$, 727 obs).] The final 2 observed seasons of each selected player are excluded from the training likelihood. The model therefore sees the player's career up to but not including their most recent 2 seasons, then is scored on those withheld seasons. This  mimics the real task of forecasting a player's next 1--2 seasons from their career to date.

  \item[Hold-out First $k$ ($k{=}2$, 757 obs).] The first 2 observed seasons of each selected player are withheld. The model sees seasons 3 onward and must recover the player's rookie and sophomore output. This tests whether the concave prior and latent embedding generalise backward in time and whether the model is well-calibrated for players with sparse early-career histories.

  \item[Random Interior (413 obs).] For each selected player, a random 20\% of their interior seasons (all but the first and last) are withheld. Unlike the last-$k$ and first-$k$ schemes, the model has observations on both sides of every gap, so this tests interpolation rather than extrapolation. 

  \item[Hold-out Peak (450 obs).] The single season of peak observed production for each selected player is withheld. This directly evaluates whether the model's explicit peak parameterisation ($t^*$, $f^*$) allows it to recover the peak even when it is unobserved.

  \item[Stratified Next $k$ ($k{=}2$, 1{,}484 obs).] A balanced forecasting scheme that controls for career stage. Players are first binned into six career-length cohorts by their maximum observed age, using the thresholds $\geq 25, 27, 29, 31, 33, 35$ (each player assigned to the highest cohort it qualifies for), with corresponding hold-out start ages of $23, 25, 27, 29, 31, 33$. For every selected player the entire career tail from its cohort's start age onward is removed from the training likelihood and predictions are scored on the next $k$ withheld seasons, i.e.\ the cohort's two-year window (ages 23--24, 25--26, $\dots$, 33--34). Unlike Hold-out Last $k$, whose scored seasons cluster at each player's individual career end, this scheme evaluates the next-two-season forecast evenly across the full $23$--$34$ age range, isolating how forecast quality varies by career stage.
\end{description}

\clearpage
\section{Latent Dimension Ablation}
\label{sec:ablation_q}

To select the latent dimension $q$ we fit the model for $q \in \{5, 10, 15, 20\}$, holding all other hyperparameters fixed, and scored each fit on held-out data under four of the holdout schemes of Appendix~\ref{sec:model_validation} (Hold-out Last $k$, Hold-out First $k$, Random Interior and Hold-out Peak).  Table~\ref{tab:ablation_rmse_holdout_last_k} reports holdout RMSE with bias for every metric and value of $q$, one block per scheme.

The ablation was carried out with MAP fits of the model under each holdout scheme rather than with full posterior sampling.

\input{figs/model_plots/coverage/ablation_holdout_body.tex}

\clearpage
\section{Additional Figures and Tables}

\begin{figure}[tbp]
    \centering
    \includegraphics[width=0.6\linewidth]{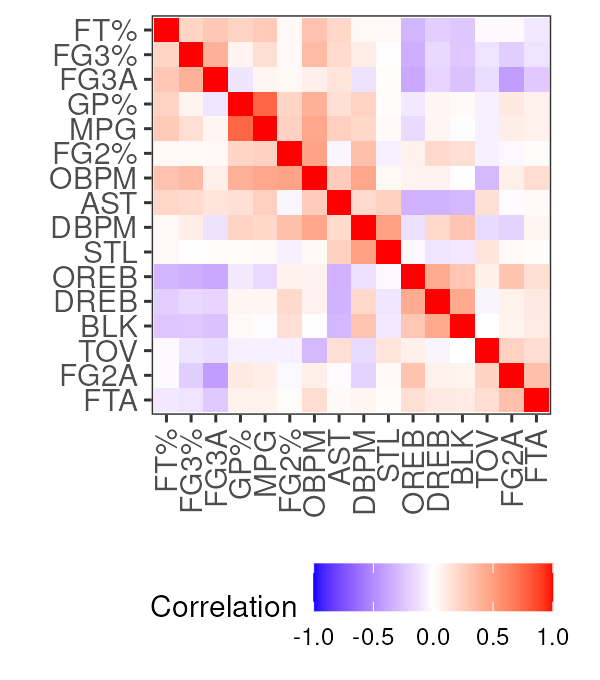}
    \caption{Between-player metric correlation.}
    \label{fig:between_cor}
\end{figure}

\begin{figure}[tbp]
    \centering
    \includegraphics[width=0.55\linewidth]{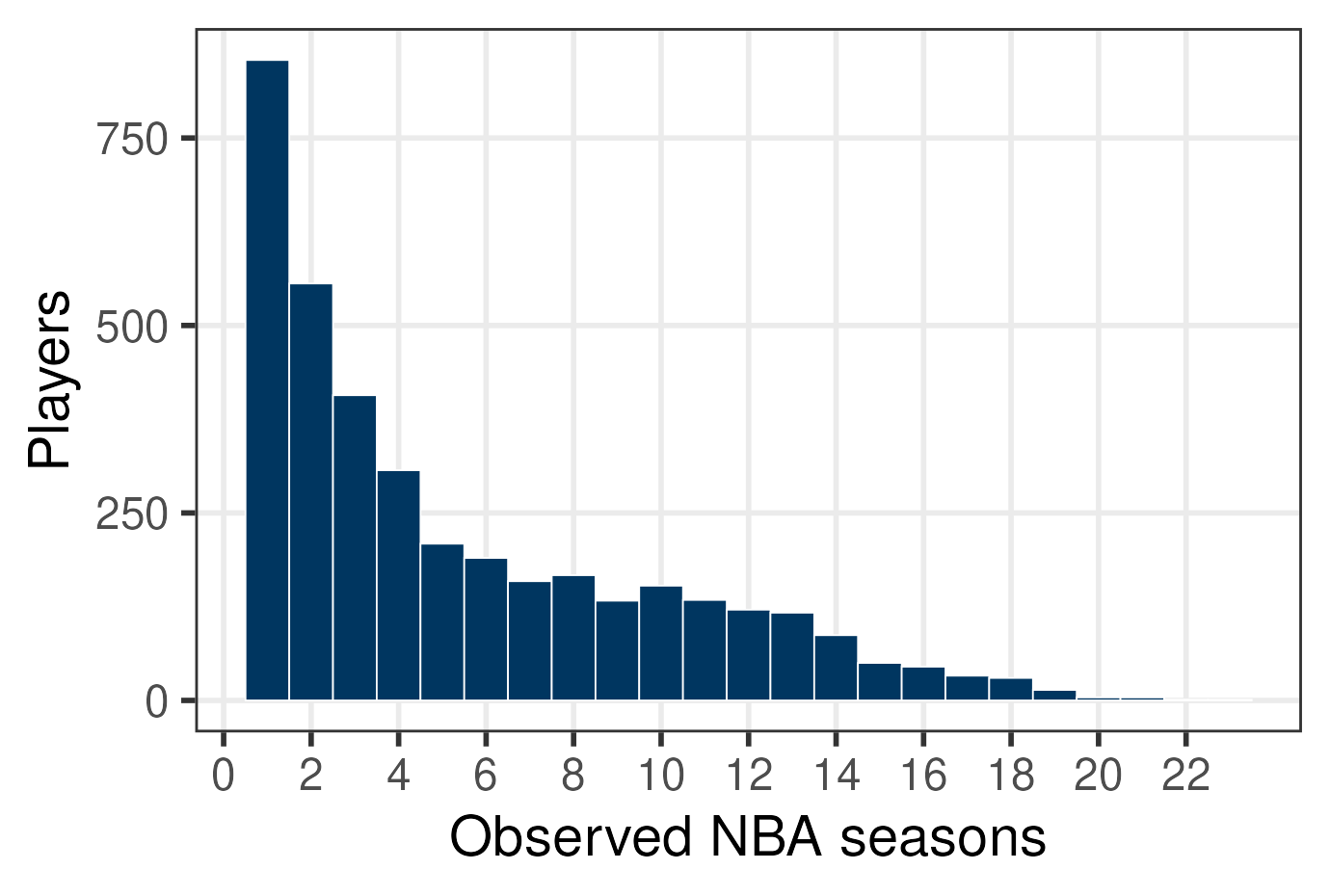}
    \caption{Number of observed NBA seasons per player in our data window. The distribution is heavily right-skewed. The mode of career lengths lasts a single season, 37\% of players are observed for two seasons or fewer, and the median career is four seasons against a maximum of 23.  Exit from the sample is strongly concentrated among short, low-production careers.}
    \label{fig:nseasons_hist}
\end{figure}

\begin{figure}[tbp]
  \centering
  \includegraphics[width=\linewidth]{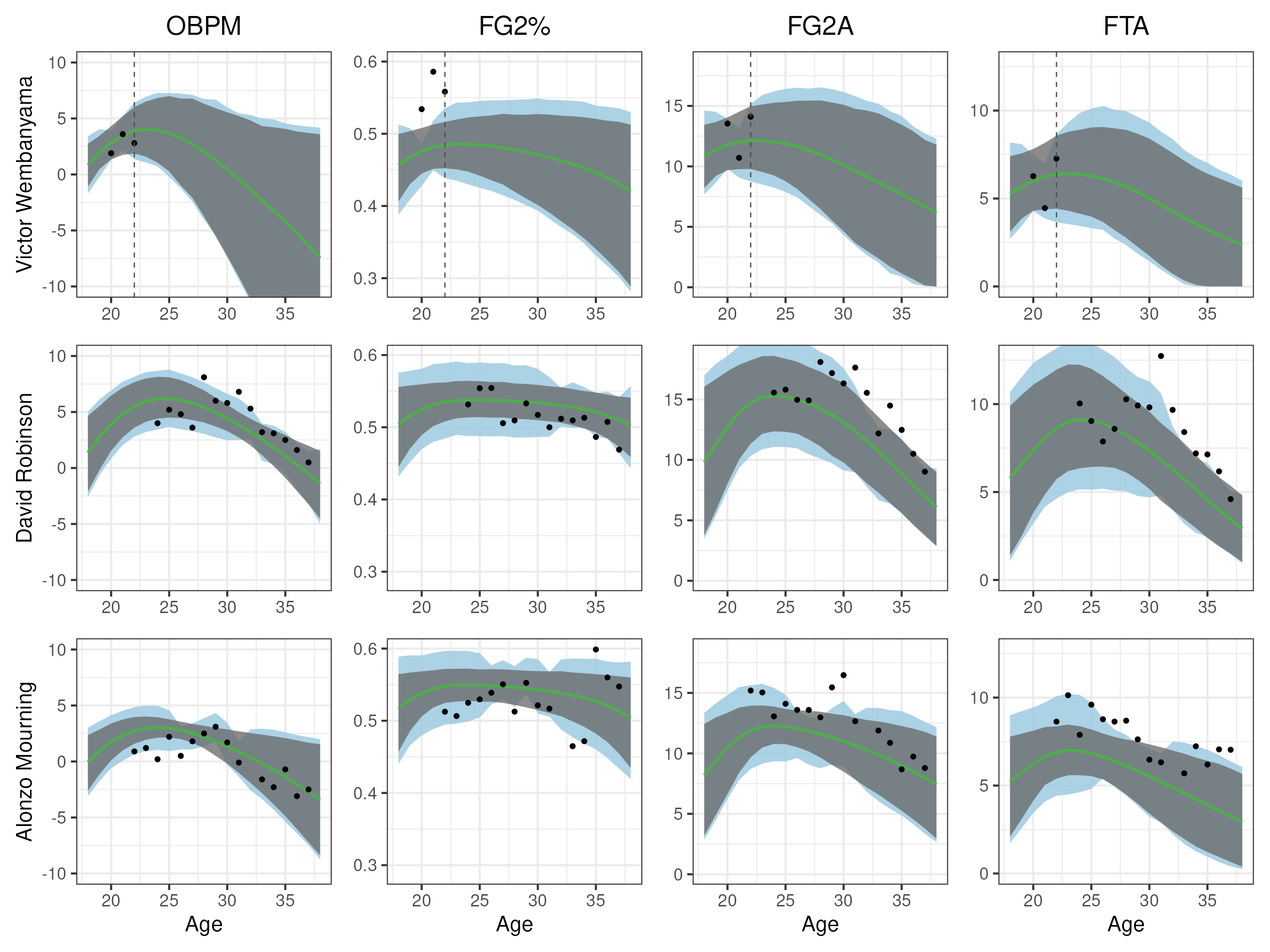}
  \caption{Posterior production curves for Victor Wembanyama (top row) and his two nearest latent-space neighbors, David Robinson (middle row) and Alonzo Mourning (bottom row), with the four representative metrics (OBPM, FG2\%, FG2A, FTA) as columns.  The dashed vertical line in the top row marks Wembanyama's last observed season; the curve beyond it is a forecast.  The posterior mean of the aging curve $f_k(X_p, t)$ is the green line, with its 95\% credible band in gray; the light-blue band is the 95\% interval for the full latent trajectory $f_k(X_p,t) + r^k_p(t)$, which adds the per-player AR(1) term. Observed values are black points.}
  \label{fig:posterior_curves}
\end{figure}

%% file: figs/model_plots/coverage/prior_spec.tex
\begin{table}[tbp]
    \centering
    \resizebox{\linewidth}{!}{%
    \begin{tabular}{llll}
        \toprule
        Parameter & Prior & Role & In MCMC \\
        \midrule
        $\sigma_W$ & $\mathrm{HalfNormal}(1)$ & Scale of the covariate projection $W$ & MAP \\
        $W$ & $\mathcal{N}(0, 1)$ & Covariate projection to the latent mean $Z_p^\top W$ & sampled \\
        $\sigma_X$ & $\mathrm{HalfNormal}(0.5)$ & Latent residual scale & MAP \\
        $\varepsilon_p$ & $t_{4}(0,\; \sqrt{1/2})$ & Latent residual, $X_p = Z_p^\top W + \sigma_X\,\varepsilon_p$ & sampled \\
        $\ell_j$ & $\mathrm{HalfNormal}(1)$ & Latent-kernel ARD lengthscale, per latent dimension & MAP \\
        $\Omega$ & $\mathcal{N}(0, 1)$ & Random Fourier frequencies & sampled \\
        $\boldsymbol\gamma^k_{\text{val}}$ & $\mathcal{N}(0, 1)$ & Peak-value weights, $f^{*}_k(X_p)$ & sampled \\
        $\boldsymbol\gamma^k_{\text{age}}$ & $\mathcal{N}(0, 1)$ & Peak-age weights, $t^{*}_k(X_p)$ & sampled \\
        $\sigma_c$ & $\mathcal{IG}(2, 1)$ & Peak-value process scale & MAP \\
        $\sigma_t$ & $\mathrm{HalfNormal}(0.5)$ & Peak-age process scale & sampled \\
        $\mathbf{A}_k$ & $\mathcal{N}(0, 1)$ & HSGP curvature weights, entries scaled by $\sqrt{S(\sqrt{\lambda_l})}$ & sampled \\
        $\alpha_k$ & $\mathrm{HalfNormal}(1)$ & Curvature amplitude & sampled \\
        $\ell_k$ & $\log\mathcal{N}(\log 3, 0.7)$ & Curvature-kernel lengthscale & sampled \\
        $\rho_k$ & $\mathrm{Uniform}(-0.5, 0.5)$ & AR(1) autocorrelation of $r^k_p(t)$ & sampled \\
        $\sigma^k_{\mathrm{AR}}$ & $\mathrm{HalfNormal}(0.3)$ & AR(1) innovation scale & sampled \\
        $r^k_p(t)$ & $\mathcal{N}(0, 1)$ & AR(1) innovations (non-centred) & sampled \\
        $r^k_p(t_0)$ & $\mathcal{N}(0, 1)$ & AR(1) initial state (non-centred) & sampled \\
        $\sigma^2$ & $\mathcal{IG}(5, 600)$ & Gaussian observation variance, $k \in \mathcal{G}$ & sampled \\
        $\tau_k^2,\ k \in \mathcal{P}$ & $\mathrm{Uniform}(0, 1)$ & Beta precision & sampled \\
        $\tau_k^2,\ k \in \mathcal{B}$ & $\mathrm{Exponential}(1)$ & Beta-binomial precision & sampled \\
        $\tau_k^2,\ k \in \mathcal{C}$ & $\mathrm{Exponential}(1)$ & Negative-binomial concentration & sampled \\
        $\eta_{\text{glob}}$ & $\mathcal{N}(\log 0.04,\; 0.5^2)$ & Global log baseline hazard & sampled \\
        $\gamma_{\text{glob}}$ & $\mathcal{N}(\log 0.15,\; 0.3^2)$ & Global log aging rate & sampled \\
        $\sigma_{\text{exit}}$ & $\mathrm{HalfNormal}(0.5)$ & Baseline-hazard latent spread & MAP \\
        $\psi_{\text{exit}}$ & $\mathcal{N}(0, 1)$ & Log baseline-hazard weights & sampled \\
        $\psi_{\text{rate}}$ & $\mathcal{N}(0, 0.1^2)$ & Log aging-rate weights & sampled \\
        \bottomrule
    \end{tabular}}%
    \caption{
    $\mathcal{IG}(\alpha,\beta)$ is an Inverse-Gamma with shape $\alpha$ and rate $\beta$;
    $t_\nu(0, s)$ is a Student-$t$ with $\nu$ degrees of freedom and scale $s$.  Weight sites are
    resolved in the random-feature space, so the peak weights are $2m = 100$-dimensional per
    metric and the curvature weights $\mathbf{A}_k$ are $2m \times M$; the latent coordinates
    themselves are $q = 10$-dimensional.  The final column records whether a site is sampled
    during MCMC or held fixed at its MAP estimate.}
    \label{tab:prior_spec}
\end{table}

%% file: figs/nba_convex_max_tvrflvm_AR_posyear_long/stratified_next_k/mcmc/plots/mcmc/convergence_table_paper.tex
\begin{table}[htbp]
    \centering\small
    \begin{tabular}{lrrr}
        \toprule
        Quantity & $n$ & median $\hat{R}$ & median tail-ESS \\
        \midrule
        $t^{*}$ (peak age) & 46,767 & 1.077 & 144 \\
        $f^{*}$ (peak value) & 46,767 & 1.029 & 983 \\
        $\mu_{ptk}$ (concave curve) & 982,107 & 1.040 & 421 \\
        $\lVert X_p \rVert$ (latent) & 2,751 & 1.049 & 367 \\
        \bottomrule
    \end{tabular}
    \caption{Convergence diagnostics over 50,400 posterior draws (4 chains). $n$ is the number of cells summarised. Entries are medians over those $n$ cells. $\hat{R}$ is the rank-normalized split-$\hat{R}$ of \citet{vehtari2021rank}. Tail-ESS is the smaller of the effective sample sizes of the indicators for the 5\% and 95\% quantiles \citep{vehtari2021rank}.}
    \label{tab:convergence}
\end{table}

%% file: figs/model_plots/coverage/obpm_ablation.tex
\begin{table}[htbp]
\centering
\small
\begin{tabular}{lc}
\toprule
Restriction imposed & Stratified Next $k$ \\
\midrule
concavity ($\mu'' \leq 0$) & $+226 \pm 73$$^{\dagger}$ \\
constant curvature per side & $-5 \pm 16$ \\
symmetry ($\alpha_u = \alpha_d$) & $-68 \pm 18$$^{\dagger}$ \\
\textbf{metric set: 17 $\to$ 1} & $-355 \pm 108$$^{\dagger}$ \\
\bottomrule
\end{tabular}
\caption{  Every entry is the change in held-out log predictive density on OBPM alone; $\dagger$ marks $|\Delta| > 1.96\,\mathrm{se}_{\Delta}$ with the paired standard error of \citet{vehtari2017practical}.  Negative means the restriction costs accuracy.  The first three rows are the curve-shape ladder of Table~2 of the main text  on OBPM observations rather than pooled across metrics.  The last row is the identical model, priors, de-trend and holdout fit to OBPM alone.  Results are for the Stratified Next $k$ scheme.}
\label{tab:obpm_ablation}
\end{table}

%% file: figs/model_plots/coverage/ablation_holdout_body.tex
\begin{table}[tbp]
\centering
\small
\resizebox{\linewidth}{!}{
\begin{tabular}{lrrrr}
\toprule
Metric & $q=5$ & $q=10$ & $q=15$ & $q=20$ \\
  & \small RMSE (bias) & \small RMSE (bias) & \small RMSE (bias) & \small RMSE (bias) \\
\midrule
GP\% & 0.348 (-0.327) & 0.365 (-0.346) & \textbf{0.324 (-0.284)} & 0.329 (-0.261) \\
USG\% & 0.055 (-0.004) & 0.052 (+0.008) & \textbf{0.049 (+0.001)} & 0.055 (-0.011) \\
MPG & 0.172 (+0.053) & \textbf{0.138 (+0.038)} & 0.138 (+0.040) & 0.140 (+0.003) \\
OBPM & 2.722 (+1.073) & \textbf{1.950 (+0.437)} & 2.299 (+0.931) & 2.300 (+0.728) \\
DBPM & 1.206 (+0.121) & 1.054 (+0.060) & \textbf{0.992 (+0.152)} & 1.050 (+0.192) \\
BLK & 0.282 (+0.241) & 0.280 (+0.232) & 0.320 (+0.221) & \textbf{0.224 (+0.173)} \\
STL & 0.871 (+0.859) & \textbf{0.768 (+0.763)} & 0.882 (+0.852) & 0.855 (+0.797) \\
AST & 1.058 (+1.052) & \textbf{0.999 (+0.919)} & 1.017 (+0.944) & 1.603 (+1.414) \\
DREB & 1.017 (+0.087) & \textbf{1.002 (+0.361)} & 1.315 (+0.376) & 1.346 (+0.322) \\
OREB & 0.632 (+0.554) & 0.616 (+0.558) & \textbf{0.500 (+0.482)} & 0.769 (+0.718) \\
TOV & 0.675 (+0.096) & \textbf{0.547 (+0.217)} & 0.626 (+0.188) & 0.601 (+0.149) \\
FTA & 2.041 (+1.823) & 1.996 (+1.904) & 2.147 (+1.927) & \textbf{1.829 (+1.661)} \\
FG2A & 3.696 (+1.350) & 2.637 (+1.008) & \textbf{2.329 (+0.775)} & 2.443 (+0.654) \\
FG3A & 1.034 (+0.636) & \textbf{0.803 (+0.446)} & 0.861 (+0.227) & 1.563 (+0.596) \\
FT\% & 0.367 (-0.117) & 0.375 (-0.104) & \textbf{0.354 (-0.092)} & 0.420 (-0.116) \\
FG2\% & 0.365 (+0.364) & \textbf{0.202 (+0.202)} & 0.406 (+0.406) & 0.497 (+0.496) \\
FG3\% & 0.326 (+0.243) & 0.323 (+0.236) & 0.328 (+0.235) & \textbf{0.304 (+0.202)} \\
\bottomrule
\end{tabular}
}
\caption{Holdout RMSE and bias (in parentheses) by metric and latent dimension $q$ --- Hold-out Last $k$ scheme. Best $q$ per metric in bold.}
\label{tab:ablation_rmse_holdout_last_k}
\end{table}

%% file: main.bbl
\begin{thebibliography}{35}
\providecommand{\natexlab}[1]{#1}
\providecommand{\url}[1]{\texttt{#1}}
\expandafter\ifx\csname urlstyle\endcsname\relax
  \providecommand{\doi}[1]{doi: #1}\else
  \providecommand{\doi}{doi: \begingroup \urlstyle{rm}\Url}\fi

\bibitem[Allen and Hopkins(2015)]{allen2015age}
Sian~V Allen and William~G Hopkins.
\newblock Age of peak competitive performance of elite athletes: A systematic
  review.
\newblock \emph{Sports Medicine}, 45\penalty0 (10):\penalty0 1431--1441, 2015.

\bibitem[Andersen et~al.(2018)Andersen, Siivola, Riutort-Mayol, and
  Vehtari]{convexityconstraints}
Michael~Riis Andersen, Eero Siivola, Gabriel Riutort-Mayol, and Aki Vehtari.
\newblock A non-parametric probabilistic model for monotonic functions.
\newblock All of Bayesian Nonparametrics Workshop (BNP@NeurIPS), 2018.

\bibitem[Aßmann et~al.(2016)Aßmann, Boysen-Hogrefe, and Pape]{procrustes}
Christian Aßmann, Jens Boysen-Hogrefe, and Markus Pape.
\newblock Bayesian analysis of static and dynamic factor models: An ex-post
  approach towards the rotation problem.
\newblock \emph{Journal of Econometrics}, 192\penalty0 (1):\penalty0 190--206,
  2016.
\newblock ISSN 0304-4076.
\newblock \doi{https://doi.org/10.1016/j.jeconom.2015.10.010}.
\newblock URL
  \url{https://www.sciencedirect.com/science/article/pii/S0304407615002626}.

\bibitem[Berry et~al.(1999)Berry, Reese, and Larkey]{berry1999bridging}
Scott~M Berry, C~Shane Reese, and Patrick~D Larkey.
\newblock Bridging different eras in sports.
\newblock \emph{Journal of the American Statistical Association}, 94\penalty0
  (447):\penalty0 661--676, 1999.

\bibitem[Bonilla et~al.(2008)Bonilla, Chai, and Williams]{NIPS2007_66368270}
Edwin~V Bonilla, Kian Chai, and Christopher Williams.
\newblock Multi-task gaussian process prediction.
\newblock In J.~Platt, D.~Koller, Y.~Singer, and S.~Roweis, editors,
  \emph{Advances in Neural Information Processing Systems}, volume~20. Curran
  Associates, Inc., 2008.
\newblock URL
  \url{https://proceedings.neurips.cc/paper/2007/file/66368270ffd51418ec58bd793f2d9b1b-Paper.pdf}.

\bibitem[Bradbury(2009)]{bradbury2009peak}
John~Charles Bradbury.
\newblock Peak athletic performance and ageing: Evidence from baseball.
\newblock \emph{Journal of Sports Sciences}, 27\penalty0 (6):\penalty0
  599--610, 2009.

\bibitem[Chessa et~al.(2023)Chessa, D’Urso, De~Giovanni,
  et~al.]{chessa2023complex}
Antonio Chessa, Pierpaolo D’Urso, Luigi De~Giovanni, et~al.
\newblock Complex networks for community detection of basketball players.
\newblock \emph{Annals of Operations Research}, 325:\penalty0 363--389, 2023.
\newblock \doi{10.1007/s10479-022-04647-x}.
\newblock URL \url{https://doi.org/10.1007/s10479-022-04647-x}.

\bibitem[Ericsson et~al.(1993)Ericsson, Krampe, and
  Tesch-R{\"o}mer]{ericsson1993deliberate}
K~Anders Ericsson, Ralf~T Krampe, and Clemens Tesch-R{\"o}mer.
\newblock The role of deliberate practice in the acquisition of expert
  performance.
\newblock \emph{Psychological Review}, 100\penalty0 (3):\penalty0 363--406,
  1993.

\bibitem[Fair(2007)]{fair2007estimated}
Ray~C Fair.
\newblock Estimated age effects in athletic events and chess.
\newblock \emph{Experimental Aging Research}, 33\penalty0 (1):\penalty0 37--57,
  2007.

\bibitem[Gundersen et~al.(2021)Gundersen, Zhang, and
  Engelhardt]{gundersen2021latent}
Gregory~W. Gundersen, Michael~Minyi Zhang, and Barbara~E. Engelhardt.
\newblock Latent variable modeling with random features.
\newblock In \emph{Proceedings of the 24th International Conference on
  Artificial Intelligence and Statistics}, volume 130 of \emph{Proceedings of
  Machine Learning Research}, pages 1333--1341. PMLR, 2021.
\newblock URL \url{https://proceedings.mlr.press/v130/gundersen21a.html}.

\bibitem[Lailvaux et~al.(2014)Lailvaux, Wilson, and
  Kasumovic]{lailvaux2014trait}
Simon~P Lailvaux, Robbie Wilson, and Michael~M Kasumovic.
\newblock Trait compensation and sex-specific aging of performance in male and
  female professional basketball players.
\newblock \emph{Evolution}, 68\penalty0 (5):\penalty0 1523--1532, 2014.

\bibitem[Lawrence(2005)]{JMLR:v6:lawrence05a}
Neil Lawrence.
\newblock Probabilistic non-linear principal component analysis with gaussian
  process latent variable models.
\newblock \emph{Journal of Machine Learning Research}, 6\penalty0
  (60):\penalty0 1783--1816, 2005.
\newblock URL \url{http://jmlr.org/papers/v6/lawrence05a.html}.

\bibitem[Macnamara et~al.(2014)Macnamara, Hambrick, and
  Oswald]{macnamara2014deliberate}
Brooke~N Macnamara, David~Z Hambrick, and Frederick~L Oswald.
\newblock Deliberate practice and performance in music, games, sports,
  education, and professions: A meta-analysis.
\newblock \emph{Psychological Science}, 25\penalty0 (8):\penalty0 1608--1618,
  2014.

\bibitem[Myers(2020)]{myers2020bpm2}
Daniel Myers.
\newblock About box plus/minus (bpm).
\newblock Basketball-Reference.com, 2020.
\newblock URL \url{https://www.basketball-reference.com/about/bpm2.html}.
\newblock Accessed: 2026-08-26.

\bibitem[Page et~al.(2007)Page, Fellingham, and Reese]{page2007using}
Garritt~L Page, Gilbert~W Fellingham, and C~Shane Reese.
\newblock Using box-scores to determine a position's contribution to winning
  basketball games.
\newblock \emph{Journal of Quantitative Analysis in Sports}, 3\penalty0 (4),
  2007.

\bibitem[Page et~al.(2013)Page, Barney, and McGuire]{page2013effect}
Garritt~L Page, Bradley~J Barney, and Aaron~T McGuire.
\newblock Effect of position, usage rate, and per game minutes played on {NBA}
  player production curves.
\newblock \emph{Journal of Quantitative Analysis in Sports}, 9\penalty0
  (4):\penalty0 337--345, 2013.

\bibitem[Phan et~al.(2019)Phan, Pradhan, and Jankowiak]{numpyro}
Du~Phan, Neeraj Pradhan, and Martin Jankowiak.
\newblock Composable effects for flexible and accelerated probabilistic
  programming in numpyro.
\newblock \emph{arXiv preprint arXiv:1912.11554}, 2019.

\bibitem[Poworoznek et~al.(2024)Poworoznek, Anceschi, Ferrari, and
  Dunson]{poworoznek2024efficientlyresolvingrotationalambiguity}
Evan Poworoznek, Niccolo Anceschi, Federico Ferrari, and David Dunson.
\newblock Efficiently resolving rotational ambiguity in bayesian matrix
  sampling with matching, 2024.
\newblock URL \url{https://arxiv.org/abs/2107.13783}.

\bibitem[Rahimi and Recht(2007)]{rffs}
Ali Rahimi and Benjamin Recht.
\newblock Random features for large-scale kernel machines.
\newblock In J.~Platt, D.~Koller, Y.~Singer, and S.~Roweis, editors,
  \emph{Advances in Neural Information Processing Systems}, volume~20. Curran
  Associates, Inc., 2007.
\newblock URL
  \url{https://proceedings.neurips.cc/paper_files/paper/2007/file/013a006f03dbc5392effeb8f18fda755-Paper.pdf}.

\bibitem[Salthouse(2009)]{salthouse2009when}
Timothy~A Salthouse.
\newblock When does age-related cognitive decline begin?
\newblock \emph{Neurobiology of Aging}, 30\penalty0 (4):\penalty0 507--514,
  2009.

\bibitem[Schulz and Curnow(1988)]{schulz1988peak}
Richard Schulz and Christine Curnow.
\newblock Peak performance and age among superathletes: Track and field,
  swimming, baseball, tennis, and golf.
\newblock \emph{Journal of Gerontology}, 43\penalty0 (5):\penalty0 P113--P120,
  1988.

\bibitem[Schulz et~al.(1994)Schulz, Musa, Staszewski, and
  Siegler]{schulz1994relationship}
Richard Schulz, Donald Musa, James Staszewski, and Robert~S Siegler.
\newblock The relationship between age and major league baseball performance:
  Implications for development.
\newblock \emph{Psychology and Aging}, 9\penalty0 (2):\penalty0 274--286, 1994.

\bibitem[Silver(2019)]{natesilver538_2019}
Nate Silver.
\newblock How our {RAPTOR} metric works.
\newblock
  \url{https://fivethirtyeight.com/features/how-our-raptor-metric-works/}, Oct
  2019.

\bibitem[Silver and Fischer-Baum(2015)]{natesilver538_2015}
Nate Silver and Reuben Fischer-Baum.
\newblock We're predicting the career of every {NBA} player. {H}ere's how.
\newblock
  \url{https://fivethirtyeight.com/features/how-were-predicting-NBA-player-career/},
  Oct 2015.

\bibitem[Simonton(1988)]{simonton1988age}
Dean~Keith Simonton.
\newblock Age and outstanding achievement: What do we know after a century of
  research?
\newblock \emph{Psychological Bulletin}, 104\penalty0 (2):\penalty0 251--267,
  1988.

\bibitem[Solin and Särkkä(2019)]{Solin_2019}
Arno Solin and Simo Särkkä.
\newblock Hilbert space methods for reduced-rank gaussian process regression.
\newblock \emph{Statistics and Computing}, 30\penalty0 (2):\penalty0 419–446,
  August 2019.
\newblock ISSN 1573-1375.
\newblock \doi{10.1007/s11222-019-09886-w}.
\newblock URL \url{http://dx.doi.org/10.1007/s11222-019-09886-w}.

\bibitem[Tanaka and Seals(2008)]{tanaka2008endurance}
Hirofumi Tanaka and Douglas~R Seals.
\newblock Endurance exercise performance in masters athletes: Age-associated
  changes and underlying physiological mechanisms.
\newblock \emph{The Journal of Physiology}, 586\penalty0 (1):\penalty0 55--63,
  2008.

\bibitem[Terner and Franks(2021)]{terner2020modeling}
Zachary Terner and Alexander Franks.
\newblock Modeling player and team performance in basketball.
\newblock \emph{Annual Review of Statistics and Its Application}, 8:\penalty0
  1--23, 2021.

\bibitem[Vaci et~al.(2019)Vaci, Coci{\'c}, Gula, and
  Bilali{\'c}]{vaci2019large}
Nemanja Vaci, Dijana Coci{\'c}, Bartosz Gula, and Merim Bilali{\'c}.
\newblock Large data and bayesian modeling—aging curves of {NBA} players.
\newblock \emph{Behavior {R}esearch {M}ethods}, pages 1--21, 2019.

\bibitem[Vehtari et~al.(2017)Vehtari, Gelman, and Gabry]{vehtari2017practical}
Aki Vehtari, Andrew Gelman, and Jonah Gabry.
\newblock Practical bayesian model evaluation using leave-one-out
  cross-validation and {WAIC}.
\newblock \emph{Statistics and Computing}, 27\penalty0 (5):\penalty0
  1413--1432, 2017.

\bibitem[Vehtari et~al.(2021)Vehtari, Gelman, Simpson, Carpenter, and
  B\"{u}rkner]{vehtari2021rank}
Aki Vehtari, Andrew Gelman, Daniel Simpson, Bob Carpenter, and Paul-Christian
  B\"{u}rkner.
\newblock Rank-normalization, folding, and localization: An improved
  $\widehat{R}$ for assessing convergence of {MCMC} (with discussion).
\newblock \emph{Bayesian Analysis}, 16\penalty0 (2):\penalty0 667--718, 2021.

\bibitem[Vinu{\'e} and Epifanio(2019)]{vinue2019forecasting}
Guillermo Vinu{\'e} and Irene Epifanio.
\newblock Forecasting basketball players' performance using sparse functional
  data.
\newblock \emph{Statistical Analysis and Data Mining: The ASA Data Science
  Journal}, 2019.

\bibitem[Vinu{\'e} et~al.(2015)Vinu{\'e}, Epifanio, and
  Alemany]{vinue2015archetypoids}
Guillermo Vinu{\'e}, Irene Epifanio, and Sandra Alemany.
\newblock Archetypoids: A new approach to define representative archetypal
  data.
\newblock \emph{Computational Statistics \& Data Analysis}, 87:\penalty0
  102--115, 2015.

\bibitem[Wakim and Jin(2014)]{wakim2014functional}
Alexander Wakim and Jimmy Jin.
\newblock Functional data analysis of aging curves in sports.
\newblock \emph{arXiv preprint arXiv:1403.7548}, 2014.

\bibitem[Zhang et~al.(2023)Zhang, Gundersen, and
  Engelhardt]{zhang2023bayesiannonlinearlatentvariable}
Michael~Minyi Zhang, Gregory~W. Gundersen, and Barbara~E. Engelhardt.
\newblock Bayesian non-linear latent variable modeling via random fourier
  features, 2023.
\newblock URL \url{https://arxiv.org/abs/2306.08352}.

\end{thebibliography}
